\documentclass[11pt]{article}
\usepackage{hyperref}
\usepackage{amsfonts,amssymb,amsmath,amsthm}
\usepackage{slashed}
\usepackage{color}
\usepackage{pdfsync}
\usepackage{cancel}
\usepackage{soul}
\definecolor{purple}{rgb}{0.6,0,0.6}
\definecolor{bl}{rgb}{0.2,0.1,0.9}
\definecolor{tur}{rgb}{0.1,0.9,0.9}
\definecolor{gr}{rgb}{0.2,0.7,0.2}
\definecolor{pink}{rgb}{1,0,0.8}

\newtheorem{theorem}{Theorem}[section]
\newtheorem{proposition}[theorem]{Proposition}
\newtheorem{definition}[theorem]{Definition}
\newtheorem{conjecture}[theorem]{Conjecture}
\newtheorem{assumptions}[theorem]{Assumption}
\theoremstyle{remark}
\newtheorem{remark}[theorem]{Remark}
\newcommand\mb[1]{\mbox{\rm#1}}
\newcommand\ee[1]{\stackrel{\mbox{\scriptsize(\ref{#1})}}{=}}
\newcommand\eet[1]{\stackrel{\mbox{\scriptsize\rm#1}}{=}}
\newcommand{\qu}{\overline}
\newcommand\tr[1][]{\mathop{\mathrm{tr}_{#1}}}
\newcommand\sgn{\mathrm{sgn}}
\newcommand\wt{\widetilde}
\newcommand{\AAA}{\mathcal A}
\newcommand{\C}{\mathbb C}
\newcommand{\E}{\mathbb E}
\newcommand{\EEE}{\mathcal E}
\renewcommand{\H}{\mathbb H}
\newcommand{\HHH}{\mathcal H}
\newcommand{\N}{\mathbb N}

\newcommand{\R}{\mathbb R}
\newcommand{\Z}{\mathbb Z}
\numberwithin{equation}{section}
\begin{document}
\title{Hodge-elliptic genera and how they govern K3 theories}
\author{\Large Katrin Wendland\footnote{katrin.wendland@math.uni-freiburg.de}\\ [5pt]
 \normalsize{$^*$Mathematics Institute, University of Freiburg}\\
 \normalsize{D-79104 Freiburg, Germany.  }}
 \date{\vspace{-5ex}}

%
%
\maketitle
\abstract{The  (complex) Hodge-elliptic genus and its conformal field theoretic counterpart were recently introduced by Kachru and Tripathy, refining the traditional complex elliptic genus. We construct a different, so-called chiral Hodge-elliptic genus, which is expected to agree with the generic conformal field theoretic Hodge-elliptic genus, in contrast to the complex Hodge-elliptic genus as originally defined.

For K3 surfaces $X$, the chiral Hodge-elliptic genus is shown to be independent of all moduli. Moreover, employing Kapustin's results on infinite volume limits it is shown that it agrees 
with the generic conformal field theoretic Hodge-elliptic genus of K3 theories, while the complex Hodge-elliptic genus
does not. This new invariant governs part of the field content of K3 theories, supporting the idea that all their spaces of states have a common subspace which underlies the generic conformal field theoretic Hodge-elliptic genus, and thereby the complex elliptic genus. Mathematically, this space is modelled by the sheaf cohomology of the chiral de Rham complex of $X$. It decomposes into irreducible representations of the $N=4$ superconformal algebra such that the multiplicity spaces of all massive representations have precisely the dimensions required in order to furnish the representation of the Mathieu group $M_{24}$ that is predicted by Mathieu Moonshine. This is interpreted as evidence in favour of the ideas of symmetry surfing, which have been proposed by Taormina and the author, along with the claim that the sheaf cohomology of the chiral de Rham complex is a natural home for Mathieu Moonshine.

These investigations also imply  that the generic chiral algebra of  K3 theories is precisely the $N=4$ superconformal algebra at central charge $c=6$, if the usual predictions on infinite volume limits from string theory hold true.}
\clearpage

\tableofcontents
\section{Introduction and summary}\label{intro}
In conformal field theory (CFT), mathematical structures that have
a counterpart in geometry play a key role. 
The success stories of orbifolding
techniques \cite{ha90,th97,dhvw85,dhvw86} and 
mirror symmetry \cite{lvw89,grpl90,cogp91,cls90,ko95}
are examples of this. 
As a common feature, there are striking correspondences
that allow to recover geometric invariants, on the one hand, 
in terms of quantum numbers in conformal field theory, 
on the other. 
Though rooted in string theory, the study of the
relevant invariants and their imprint in conformal
field theory
is a fruitful mathematical enterprise, independently of string theory.
\smallskip

A concrete example of such a correspondence identifies
the Euler characteristic of a compact Calabi-Yau manifold $X$
with the Witten index  of
an associated conformal field theory, obtained
as non-linear sigma model on $X$ \cite{wi82}. 
Indeed, the predictive power of string theory
motivates this correspondence
\cite{wi82,lvw89},
since a large volume limit of the sigma model
is expected to recover the cohomology of the 
target manifold $X$.
More generally, 
the complex elliptic genus
\cite{hi88,wi88,kr90} of $X$, which can 
be defined as the holomorphic Euler characteristic
of a certain virtual bundle $\E_{q,-y}$ on $X$, is expected
to be recovered from the part of the partition function
of the associated superconformal field theory
which a topological half-twist \cite{egya90,wi91} projects to 
\cite{akmw87,eoty89,dfya93,wi94}.

Recently, Kachru and Tripathy have defined 
a very interesting refined version of the complex elliptic genus
of a compact Calabi-Yau manifold $X$, 
which they call the {\em Hodge-elliptic genus}
\cite{katr16}. The key idea is to introduce an additional parameter 
which keeps track of the grading on the cohomology
of the virtual bundle $\E_{q,-y}\longrightarrow X$ that
underlies the complex elliptic genus. 
This grading, in turn, has as its counterpart the natural grading of the space of states
of an associated superconformal field theory
by the right-moving $U(1)$ charge. Indeed,
the Hodge-elliptic genus has a natural
counterpart in conformal field theory, 
which is also introduced in \cite{katr16},
and which we call
the {\em conformal field theoretic Hodge-elliptic genus}.
It is   important
to keep these two versions of the Hogde-elliptic
genus apart, as indeed they disagree,
most of the time. 
To make the difference clearer, we will add the adjective 
{\em complex} to Kachru and Tripathy's {\em Hodge elliptic
genus}.

Both these new quantities are 
very promising, since the additional grading
eliminates the typical cancellations that make
it  so difficult to reconstruct  data from 
the complex elliptic genus and its conformal
field theoretic counterpart. But
preventing such cancellations, in general,
causes a dependence on the moduli. 
Indeed, the complex Hodge-elliptic genus of a general compact Calabi-Yau manifold
$X$ in dimensions greater than two
is expected to depend on the complex structure of $X$,
while it is independent of the K\"ahler 
structure. We thus obtain meaningful geometric
invariants by insisting on a fixed complex structure
for any of our Calabi-Yau manifolds.
For complex tori and for K3 surfaces,
however, the complex Hodge-elliptic genus is independent
of the complex structure \cite{katr16}.

The conformal field theoretic Hodge-elliptic genus,
on the other hand, always severely depends on 
{\em all} the moduli. 
By definition, it is a power series 
in three formal variables with integral 
exponents. Its coefficients, up to signs, 
are just the dimensions of eigenspaces of certain
natural linear operators on the space of states
of the conformal field theory. The moduli dependence
is thus reflected in a {\em jumping} behaviour of 
the coefficients. In particular, each such coefficient
{\em generically} attains its minimal value, on 
the moduli space, while it may jump to a higher
value at non-generic points. To obtain
a quantity that stands any chance of relating
to some invariant geometric counterpart, instead 
of the conformal field theoretic Hodge-elliptic genus,
one should therefore consider  
a {\em generic conformal field theoretic 
Hodge-elliptic genus}. We shall define such a quantity, below,
as is certainly in the spirit of \cite{katr16}.
In string theory language, 
our definition amounts to an {\em infinite volume limit} 
of the conformal field theoretic Hodge elliptic
genus. 
However, we claim that
the generic conformal field theoretic Hodge-elliptic
genus of a CFT with geometric
interpretation on some Calabi-Yau manifold $X$ also
differs from the complex Hodge-elliptic genus of $X$,
in general.

The reason for this discrepancy lies in the very
definition of the complex Hodge-elliptic genus for a compact Calabi-Yau
manifold $X$ by means of the cohomology of the virtual
bundle $\E_{q,-y}\longrightarrow X$. Indeed, there
is no reason to expect the cohomology of this bundle 
to describe conformal field theory data, not even in
an infinite volume limit. According to Kapustin's
seminal insights \cite{ka05}, such an infinite
volume limit of a topological half-twist of  a sigma
model on $X$ should yield the sheaf cohomology of the
{\em  \mb{(}holomorphic\mb{)} chiral de Rham complex} $\Omega_X^{\rm ch}$
on $X$, instead.
The sheaf of sections of the virtual bundle $(-y)^{D\over2}\E_{q,y}$
is isomorphic to the {\em graded object} of the 
chiral de Rham complex for a natural filtration 
of $\Omega_X^{\rm ch}$. Therefore, the graded
Euler characteristics of the two agree, but the 
separate degrees $H^j(X,\E_{q,-y})$ of the cohomology,
which enter crucially
in the definition of the complex Hodge-elliptic genus, 
need not be isomorphic to those of $\Omega_X^{\rm ch}$.

It is therefore natural to replace the complex Hodge-elliptic genus
by a new invariant which is directly obtained from the sheaf cohomology
of the (holomorphic) chiral de Rham complex of $X$, by introducing an additional
grading. That this is possible is shown in this note, and
we call the resulting quantity the {\em chiral Hodge-elliptic genus}.
By construction, it does not depend on the K\"ahler structure of 
$X$, but it may depend on the complex structure. 
In general, we expect
the chiral Hodge-elliptic genus of any compact Calabi-Yau
manifold $X$ to agree with the generic conformal field theoretic
Hodge-elliptic genus of CFTs with geometric interpretation by
$X$ at some fixed complex structure. This allows to predict the 
behaviour of the chiral Hodge-elliptic genus under mirror symmetry.
\medskip

The new quantities introduced so far turn out to be 
particularly useful in the context of K3 surfaces
and K3 theories. Note that in the K3 setting, 
the notions of moduli spaces of CFTs and geometric interpretations 
are well understood \cite{se88,ce91,asmo94,nawe00},
even if these concepts may
seem a bit vague for other Calabi-Yau manifolds $X$.

Similarly to the complex Hodge-elliptic genus of \cite{katr16},
the chiral Hodge-elliptic genus of a K3 surface $X$
is independent of the complex 
structure, as is shown in this note by an explicit calculation. 
Under the assumption that the generic chiral algebra of all K3 
theories is the $N=4$ superconformal algebra at central
charge $c=6$, in addition, we show 
that as expected, 
the generic conformal field theoretic Hodge-elliptic genus
of K3 theories agrees with the 
chiral Hodge-elliptic genus of K3 surfaces, while
it disagrees with the complex Hodge-elliptic genus of \cite{katr16}.
Note that  a failure of the  above assumption on the generic
chiral algebra of K3 theories would be highly interesting 
in itself. Developing the 
representation theory of the relevant enhancement of the 
$N=4$ superconformal algebra would amount to a major 
advance towards the construction of K3 theories beyond the families
of examples that are known, so far. However, 
the investigation of Hodge-elliptic genera on K3 produces 
further pieces of evidence in favour of the expectation
that the generic chiral algebra of K3 theories is {\em not} 
extended beyond the $N=4$ superconformal algebra.
Indeed, in this note we show that this already follows
from the claims made in Kapustin's work \cite{ka05} on
infinite volume limits.

Let us therefore
assume, temporarily, 
that the generic chiral algebra of K3 theories 
is {\em not} enhanced beyond the $N=4$ superconformal
algebra at central charge $c=6$.  
That the chiral Hodge-elliptic genus of K3 surfaces is an
invariant which agrees with the generic conformal field theoretic
Hodge-elliptic genus of K3 theories then
means that certain quantities of our K3 theories are protected
by this invariant. Concretely, one may work under
the assumption that all K3 theories share a 
common space $\widehat\H^{\rm R}$ of protected states, 
all of which are Ramond ground states
with respect to the right-moving 
superconformal algebra. At this level, $\widehat\H^{\rm R}$ 
is just a subrepresentation of the Ramond sector of our theory
under the action of the $N=4$ superconformal algebra at
central charge $c=6$, extended by the zero mode of the 
$U(1)$ current of the right-moving superconformal algebra. 
It is left for future work to equip it with further
structure and to decide whether or not this smoothly
varies with respect to the moduli. 
The sheaf cohomology of the (holomorphic) chiral de Rham
complex of a K3 surface $X$
serves as a model for the  subspace of 
protected states which is related to $\widehat\H^{\rm R}$ by
spectral flow. Moreover, by explicitly determining
the generic conformal field theoretic Hodge-elliptic genus, we can show
that as a representation of 
the $N=4$ superconformal algebra,
$\widehat\H^{\rm R}$ splits into a direct sum of irreducible
representations whose multiplicity
spaces have precisely the
right dimensions to accomodate Mathieu Moonshine
according to \cite{eot10,ch10,ghv10a,ghv10b,eghi11,ga12}. In particular, the
multiplicity spaces of massive representations
are never virtual, as is required by the results of \cite{ga12}.
Using the sheaf cohomology
of the chiral de Rham complex of $X$ as a 
mathematically more established model, this means that the 
latter is  the 
natural home for {\em Mathieu Moonshine},
 in accord with \cite[\S4.2]{we14}.
In fact, this implication  holds  true independently of the 
above assumption on the generic chiral algebra
of K3 theories.
As such, our final conclusion agrees with 
results of  Bailin Song in \cite{so17}\footnote{The preprint \cite{so17} reached
me during the final stages of writing this note.
In fact, I am grateful to Bailin Song for his comments
on an earlier version of this work, as they allowed me
 to vastly expand the interpretation of
my results.}.

These findings support the idea of {\em symmetry surfing},
as proposed  by Taormina and the author\footnote{Occasionally, 
this idea is attributed to the seminal paper
\cite[p.~4]{eot10}, where indeed, the authors ask ``{\em Is  
it possible that these automorphism groups at 
isolated points in the moduli space of 
K3 surface are enhanced to 
$M_{24}$ over the whole of moduli space 
when we consider the elliptic genus?}''. 
Similarly, in \cite[p.~2]{ghv12}, one finds the statement
``{\em \ldots the elliptic genus is independent 
of the specific point in the moduli space of K3 that is 
considered, and thus the symmetries of the elliptic 
genus are in some sense the union of all symmetries 
that are present at different points in moduli space.}''
These  statements, however, do not anticipate any
 concrete constructions, let alone the symmetry surfing proposed in
\cite{tawe11,tawe12,tawe13}.}, starting in \cite{tawe11,tawe12}.
In particular, in \cite{tawe13}, we show how to construct
the leading order massive representation of Mathieu Moonshine
for a maximal subgroup of $M_{24}$, implementing a twist.
Our construction is based on a space
of states that is common to all K3 theories obtained
by a standard $\Z_2$-orbifolding from a theory with
target a complex torus, which we now may
recover as subspace of the
protected space $\widehat\H^{\rm R}$, above. Further
evidence in favour of this idea is provided in  \cite{gakepa16}.
\medskip

Altogether, the chiral Hodge-elliptic genus of K3 turns out to be
a  surprisingly useful new invariant. 
It truly refines the complex elliptic genus to a three-parameter
function in $(\tau,\,z,\,\nu)\in{\mathfrak H}\times\C^2$ which is
still elliptic in $z$ with respect to $\Lambda_\tau=\Z\tau+\Z$, and
which is Mock modular in $\tau$. The behaviour in the new parameter
is polynomial in $u=\exp(2\pi i\nu)$ and $u^{-1}$. 
On the other hand,  the meaning of the complex
Hodge-elliptic genus of \cite{katr16}  for conformal field theory is 
not so obvious.  Our findings may well bear some implications on the
black hole counting formulas proposed in \cite{katr16}, which were built
on the  assumption that the complex Hodge elliptic genus agrees with
the generic conformal field theoretic Hodge-elliptic genus, which however 
contradicts Kapustin's results on large volume limits, as we show. Nevertheless,
we expect the complex Hodge-elliptic genus of \cite{katr16} to 
be just as useful in geometry. We therefore also provide an implicit formula 
for the complex Hodge-elliptic genus of K3 surfaces, below.
\medskip

\noindent
In more detail, the structure of this note is as follows:

Sect.~\ref{setup} introduces the 
various inhabitants  of our zoo of Hodge-elliptic genera. 
We begin by recalling the definitions of the 
{\em conformal field theoretic} elliptic and  Hodge-elliptic 
genera in Sect.~\ref{CFTgenera}, and of 
their supposed {\em geometric} counterparts, namely
the {\em complex} elliptic  and  Hodge-elliptic genera for
compact Calabi-Yau manifolds in  Sect.~\ref{geometricellgen}. 
In Sect.~\ref{genericity}, 
we introduce the {\em generic} conformal field theoretic 
Hodge-elliptic genus, and we explain why it stands a chance
to have a geometric counterpart. 
We then introduce the  {\em chiral} Hodge-elliptic genus,
after recalling some properties of the (holomorphic) chiral de Rham complex,
which are needed for the definition of this final 
member of the Hodge-elliptic species. From this, we justify
our expectation that the {\em chiral} Hodge-elliptic genus should 
play a more important role for conformal field theory than the 
{\em complex} Hodge-elliptic 
genus. Indeed, we argue that string theory predicts that 
the generic conformal field theoretic Hodge-elliptic
genus agrees with the chiral Hodge-elliptic genus,
and not with the complex Hodge-elliptic genus, in general.
\smallskip

The following Sect.~\ref{warm} serves as a warmup for 
more serious calculations of Hodge-elliptic genera: 
for complex tori, we 
calculate each of the elliptic and Hodge-elliptic genera 
that were introduced in Sect.~\ref{setup}.
This is a straightfoward yet rewarding exercise, as it
turns out that the generic conformal field theoretic Hodge-elliptic 
genus, the chiral Hodge-elliptic genus, as well as the complex
Hodge-elliptic genus all agree in this case. This may 
explain why on first sight, one wouldn't distinguish
between the complex Hodge-elliptic genus and its chiral relative.
\smallskip

Sect.~\ref{proposal} is primarily devoted to the Hodge-elliptic
genera for K3 surfaces and K3 theories and 
contains the main results of this work.
As a first step, in Sect.~\ref{Kummer}, we calculate the 
conformal field theoretic Hodge-elliptic genus for 
standard $\Z_2$-orbifolds of non-linear sigma models
with target a complex torus. This, again, is an easy 
exercise, which we find worthwhile in view of a comparison
to the results of \cite{katr16}. 
Sect.~\ref{genericK3CFT}
addresses the generic conformal field theoretic Hodge-elliptic
genus of K3 theories. We derive a closed formula, which 
yields the latter Hodge-elliptic genus if and only if the generic chiral algebra of
K3 theories is precisely the $N=4$ superconformal algebra
at central charge $c=6$. 
Under the assumption that
this latter condition holds, in Sect.~\ref{HodgeK3},
we prove that in the K3 case, the complex Hodge-elliptic genus {\em differs}
from the generic conformal field theoretic Hodge-elliptic genus.
We also provide an implicit formula for the complex Hodge-elliptic genus 
of K3 surfaces. 

Sect.~\ref{comments} addresses the chiral
Hodge-elliptic genus of K3 surfaces. We argue that the results
up to this point yield a string theory proof of an explicit
formula for the chiral
Hodge-elliptic genus of K3 surfaces,
relying on the assumption that the generic chiral algebra of 
K3 theories is as stated above, and that an infinite volume
limit of the topological half-twist of K3 theories yields
the (sheaf) cohomology of the (holomorphic) chiral de Rham complex,
as claimed by Kapustin \cite{ka05}. We also 
provide a direct proof of this formula, which crucially uses Bailin Song's result on the
global holomorphic sections of the chiral de Rham complex 
\cite[Thm.~1.2]{so16}, thus supporting the belief that the
string theory assumptions stated above hold true.
Reversing
the argument, one obtains a string theory proof of 
Bailin Song's beautiful result
that the global holomorphic sections of the chiral de Rham complex on
a K3 surface precisely 
yield an $N=4$ superconformal vertex operator
algebra at central charge $c=6$. 
Moreover, under the assumption that this vertex operator
algebra yields the generic chiral algebra of K3 theories,
it follows that the generic conformal field theoretic Hodge-elliptic
genus of K3 theories agrees with the chiral Hodge-elliptic genus
of K3 surfaces. That the assumption holds, in turn, is found to follow from
Kapustin's claims on infinite volume limits \cite{ka05}.

We conclude in a final
Sect.~\ref{Mathieu}, explaining the consequences of our calculations
for  {\em Mathieu Moonshine}: we interpret the results of this note
as strongly supporting the expectation  \cite[\S4.2]{we14}
that the chiral de Rham complex might bear the key to understanding
Mathieu Moonshine by means of symmetry surfing, as in
\cite{tawe11,tawe12,tawe13,gakepa16}. 
It serves as a mathematical model of a common subspace of the 
space of states of all K3 theories, which is protected by the 
chiral Hodge-elliptic genus and which naturally carries the action of
an $N=4$ superconformal algebra at central charge $c=6$,
and of the finite symplectic automorphism groups of 
all K3 surfaces.
\smallskip

An Appendix lists the relevant formulae for Jacobi theta functions
and characters of the irreducible representations of the (small)
$N=4$ superconformal algebra at central
charge $c=6$.
\section{The setup: complex and Hodge-elliptic genera}\label{setup}
In this section, we recall the setup and definitions of complex and 
Hodge-elliptic genera in the conformal field theoretic as well as
the geometric context, 
and we extend these notions by a few further ideas. We begin by 
introducing some notations, and by
stating the general assumptions that are made throughout this note.
See, for example, \cite{we14} for a recent review of the relevant notions,
adapted to the applications that we have in mind, here.
\smallskip

Some of the definitions and statements given in this note hold in 
greater generality than claimed. However, to keep the exposition
more accessible,
throughout this note, we restrict our attention to a certain type of
two-dimensional Euclidean unitary superconformal field 
theories (SCFTs):
\begin{assumptions}\label{SCFTass}
In this note, by a {\em superconformal field theory (SCFT)}, a 
two-dimensional Euclidean unitary superconformal field theory with $N=(2,2)$ 
worldsheet supersymmetry is meant, througout.
For the generators
of the two commuting copies of $N=2$ super-Virasoro algebras acting on
the space of states $\H$ of such a theory,
we use the standard notations and normalizations \mb{(}see, for example,
\mb{\cite[\S2]{lvw89})}.

Furthermore, we assume that the
central charges of the theory obey
$c=\overline c$, $c=3D$, $D\in\N$, and that {\em space-time supersymmetry}
holds. 
For the space of states $\H$, we assume $\H=\H^{\rm NS}\oplus\H^{\rm R}$, with
$\H^{\rm NS}$ denoting the {\em Neveu-Schwarz sector} and $\H^{\rm R}$ the
{\em Ramond sector}, referring to left- and right-moving boundary
conditions simultaneously, as we require that the $\mb{NS-R}$ and $\mb{R-NS}$
sectors are trivial. 
We finally assume that all eigenvalues of
$J_0$ and $\qu J_0$ on $\H^R$ belong to ${c\over2}+\Z$.

We set $q:=\exp(2\pi i\tau)$ for $\tau\in\mathbb C,\,\Im(\tau)>0$, and 
$y:=\exp(2\pi i z)$ for $z\in\mathbb C$, and we denote the standard partition
function along with its $\wt{\rm R}$-sector by
\begin{eqnarray*}
Z(\tau,z) &=& 
\tr\nolimits_{\H} \left( {\textstyle{1\over2}}\left( 1+(-1)^{J_0-\qu J_0} \right)
y^{J_0} q^{L_0-{c\over24}} \qu y^{\qu J_0} \qu q^{\qu L_0-{\qu c\over24}}\right),\\
Z_{\wt{\rm R}}(\tau,z) &=& 
\tr\nolimits_{\H^R} \left( (-1)^{J_0-\qu J_0} y^{J_0} q^{L_0-{c\over24}} \qu y^{\qu J_0} \qu q^{\qu L_0-{\qu c\over24}}\right).
\end{eqnarray*}
\end{assumptions}
The assumption of space-time supersymmetry ensures that 
the linear operator $J_0-\qu J_0$ on $\H$
possesses only integral eigenvalues and
that the Ramond and the Neveu-Schwarz sector are related by 
{\em spectral flow}.  By the additional assumptions
on the spectra of $J_0$ and $\qu J_0$ on the Ramond sector,
this implies that
the eigenvalues of these  two operators are integral 
on the Neveu-Schwarz sector.
In a string theory interpretation, the properties listed in
Assumption \ref{SCFTass} are expected to be necessary to allow an
interpretation of the theory as the internal CFT of a non-linear sigma model with a 
compact $D$-dimensional Calabi-Yau target space. 
It is worth noting, however,
that these assumptions are stated here, and turn out to be useful, 
independently of such a string theory interpretation.
\smallskip

We are now ready to define and discuss the various versions
of conformal field theoretic and geometric elliptic genera that 
are the topic of this note:
\subsection{Conformal field theoretic elliptic genera}\label{CFTgenera}
In this section, we recall the definitions of the {\em conformal field theoretic 
elliptic genus} and of its refinement to the {\em conformal field theoretic
Hodge-elliptic genus}, and we briefly discuss some of their properties, in particular
explaining the necessity to require  Assumption \ref{SCFTass}.
\begin{definition}\label{cftellgen}
Consider an $N=(2,2)$ superconformal field theory at central charges
$c,\,\qu c$ and with space of states $\H=\H^{\rm NS}\oplus\H^{\rm R}$, 
which obeys  Assumption \mb{\ref{SCFTass}}.
Then
\begin{eqnarray*}
\EEE^{\rm CFT}(\H;\tau,z) 
&:=& \tr\nolimits_{\mathbb H^{\rm R}}\left( (-1)^{J_0-\overline J_0}y^{J_0} q^{L_0-{c\over24}} 
\overline q^{\overline L_0-{\overline c\over24}}\right)
\end{eqnarray*}
is the \textsc{conformal field theoretic elliptic genus} of the theory.
Now let 
$$
\widetilde\H^{\rm R}:=\left\{\phi\in\H^{\rm R} \mid \qu L_0\phi = {\textstyle{\qu c\over24}}\phi \right\}
$$ 
denote the subspace 
of the Ramond sector given by those states which
are {\em Ramond ground states} with respect to
the right-moving superconformal algebra, and set
$u:=\exp(2\pi i\nu)$ for $\nu\in\mathbb C$. 
Then, following \mbox{\rm\cite{katr16}},
$$
\EEE_{\rm Hodge}^{\rm CFT}(\H;\tau,z,\nu) 
:= \tr\nolimits_{\widetilde\H^{\rm R}}\left( (-1)^{J_0-\overline J_0}y^{J_0} u^{\overline J_0}
q^{L_0-{c\over24}} \right)
$$
is the \textsc{conformal field theoretic Hodge-elliptic genus} of the theory.
\end{definition}
\noindent
Note that $\widetilde\H^{\rm R}$ is isomorphic to the space of states of
the topological half-twist \cite{egya90,wi91}
of our SCFT, by construction:
with $Q:=\qu G_0^+$, we have $\widetilde\H^{\rm R}=\ker Q\cap\ker Q^\dagger$,
the space of ``harmonic representatives'' of the BRST cohomology
$\ker Q/\mb{im } Q$.

By the standard arguments for cancellations due
to supersymmetry,
$$
\EEE^{\rm CFT}(\H;\tau,z) 
= \tr\nolimits_{\widetilde\H^{\rm R}}\left( (-1)^{J_0-\overline J_0}y^{J_0} q^{L_0-{c\over24}}\right),
$$
showing that indeed, the conformal field theoretic Hodge-elliptic genus
$\EEE_{\rm Hodge}^{\rm CFT}(\H;\tau,z,\nu)$ 
is a refinement of the conformal field theoretic elliptic
genus
$\EEE^{\rm CFT}(\H;\tau,z)$:
\begin{equation}\label{CFTHEGisCFTEG}
\EEE_{\rm Hodge}^{\rm CFT}(\H;\tau,z,\nu=0) = \EEE^{\rm CFT}(\H;\tau,z).
\end{equation}
Let us remark that the very Assumption \ref{SCFTass}
on the type of superconformal field theories that enter the Def.~\ref{cftellgen}
are necessary, in order to
ensure that the conformal field theoretic Hodge-elliptic genus possesses a power series expansion
(with integral exponents only) in $q,\, y^{\pm{1\over2}}$, and $u^{\pm{1\over2}}$. 
For $y^{\pm{1\over2}}$ and $u^{\pm{1\over2}}$, this is immediate. For $q$, this claim
is a consequence of well-known properties of the spectral flow
(see, for example, \cite{se86,se87} or \cite[\S3.4]{gr97}), along the lines
of an argument already presented in \cite[\S3.2]{tawe17}: indeed, 
our assumptions ensure that every common eigenstate 
$\phi\in\widetilde\H^R$ of $J_0,\,\qu J_0$
and $L_0$ is the image, under spectral flow, of some state 
in the Neveu-Schwarz sector with conformal weights 
$(h;\qu h)$ and $U(1)$ charges $(Q;\qu Q)$ with 
$h\geq{|Q|\over2}$,\;
$2\qu h=\pm\qu Q$,
$h-\qu  h\in{1\over2}\Z$
and $Q,\,\qu Q\in\Z$. Moreover,  by our assumption of space-time supersymmetry,
$Q-\qu Q\in2\Z$ if and only if  $h-\qu  h\in\Z$,
or equivalently, if and only if the state is bosonic. The holomorphic conformal weight of
$\phi$ thus is
$$ 
\textstyle
h\mp{Q\over2}+{c\over24} = h-\qu h \mp{Q-\qu Q\over2} +{c\over24} 
\qquad\in\; {c\over24} +\N,
$$
as claimed. Note that the additional assumption on the eigenvalues of $J_0$ and $\qu J_0$
in Def.~\ref{cftellgen} is sometimes only tacitly made, in the literature. 
As pointed out, for example, in 
\cite[\S3.4]{gr97} and \cite[\S3.1.1]{diss}, it is equivalent to the requirement 
that the theory is invariant under the {\em purely holomorphic} and {\em anti-holomorphic
two-fold spectral flows}. By \cite{eoty89}, this condition should hold for all SCFTs that arise
as non-linear sigma models with a compact Calabi-Yau target space.

The very fact that
$\EEE_{\rm Hodge}^{\rm CFT}(\H;\tau,z,\nu)$ has a formal power series expansion in 
$q,\, y^{\pm{1\over2}}$, and $u^{\pm{1\over2}}$, 
where only integral exponents occur, seems to be the main advantage of this new quantity over the 
partition function of the underlying SCFT, as we shall see below. 
Indeed, the advantage of introducing $\EEE_{\rm Hodge}^{\rm CFT}(\H;\tau,z,\nu)$  
is not at all immediate, because just like the partition function, 
the conformal field theoretic Hodge-elliptic genus severely depends on the 
moduli of the SCFT chosen at the outset. 
In contrast, the traditional conformal field theoretic elliptic genus
$\EEE^{\rm CFT}(\H;\tau,z)$ is invariant under deformations
of the theory within the space of SCFTs of the type specialized to, 
above (see  \cite{akmw87,eoty89,dfya93,wi94} for the original
results and e.g.\ \cite[\S3.1]{diss} for a summary, including proofs).
\medskip

If a SCFT that obeys  Assumption \ref{SCFTass}
arises as a non-linear sigma model with some compact
Calabi-Yau target space $X$, then one expects 
$\EEE^{\rm CFT}(\H;\tau,z)$
to agree with the complex elliptic genus of $X$, a topological 
invariant which generalizes the {\em Hirzebruch $\chi_y$ genus} 
to a complex elliptic genus \cite{hi88,wi88,kr90} and whose definition we
recall below, see Def.~\ref{geoell}. This motivates
the following definition, which is advocated, for example,
in \cite{nawe00,we14}:
\begin{definition}\label{defk3cft}
A superconformal field theory is called a \textsc{K3 theory}, if the following
conditions hold: the SCFT is an $N=(2,2)$ superconformal field theory at
 central charges $c=6,\, \overline c=6$ with space-time supersymmetry, 
all the eigenvalues of 
the operators $J_0$ and  $\overline J_0$ on the space
of states $\H$ are integral,
and the conformal field theoretic elliptic genus of the theory
agrees with the complex elliptic genus of a K3 surface,
\begin{equation}\label{K3ellgen}
\EEE^{\rm CFT}(\H;\tau,z) 
= 8\left( \left( {\vartheta_2(\tau,z)\over \vartheta_2(\tau,0)}\right)^2
+ \left( {\vartheta_3(\tau,z)\over \vartheta_3(\tau,0)}\right)^2
+ \left( {\vartheta_4(\tau,z)\over \vartheta_4(\tau,0)}\right)^2 \right),
\end{equation}
where here and in the following, we use the standard Jacobi theta
functions $\vartheta_k(\tau,z),\; k\in\{1,\ldots,4\}$, c.f.~Appendix 
\mb{\ref{N4characters}}.
\end{definition}
Possibly, every K3 theory allows a non-linear sigma model interpretation
with target given by  some K3 surface,  but a proof is far out of reach. 
For these SCFTs, it is not hard to see and well-known to the experts\footnote{See
\cite[\S3]{we14} for a recent review, adapted to our applications.}
that the assumptions on the representation
content of such a K3 theory guarantee extended $N=(4,4)$
supersymmetry\footnote{More precisely, the relevant 
left- and right-moving superconformal 
algebras both yield a {\em small} $N=4$ superconformal algebra
according to 
\cite{aetal76}, which for simplicity, in this note, we  call 
{\em the $N=4$ superconformal algebra}.}, resonating with the
fact that every K3 surface is hyperk\"ahler. Under an assumption 
on the integrability of certain deformations, which can be justified in string theory and which
is demonstrated to all orders of perturbation theory in 
\cite{di87}, and based on the previous results \cite{se88,ce91},
the {\em moduli space of K3 theories} 
has been determined in \cite{asmo94,nawe00}. The results
of \cite{asmo94} allow to give geometric interpretations
by K3 surfaces to each theory that is accounted for in this moduli space.
Vice versa, this provides a map that assigns a unique K3 theory to every
choice of geometric moduli, in terms of a hyperk\"ahler structure,
volume and B-field on a K3 surface. To make this map explicit  is a wide open
problem, to date.
\medskip

The expectation that the conformal field theoretic elliptic genus for a non-linear sigma 
model on some Calabi-Yau manifold $X$ should agree with the 
complex elliptic genus of $X$ is rooted in the expected large volume 
behaviour of such sigma models. Following \cite{wi82,lvw89},  
the  space of Ramond ground states 
in such a large volume limit should recover the Dolbeault 
cohomology  of $X$. The work of Kapustin \cite{ka05}
shows how this expectation extends to the additional
states that are accounted for by the complex elliptic genus,
as we shall recall in Sect.~\ref{genericity}.
Since the conformal field theoretic 
elliptic genus remains invariant under deformations 
within the moduli space of theories restricted to in 
Def.~\ref{cftellgen}, the agreement must hold away from
large volume limits, as well.
\subsection{Geometric elliptic genera}\label{geometricellgen}
While the conformal field theoretic Hodge-elliptic genus does not define an invariant
under deformations of SCFTs, in any large-volume 
limit, according to \cite{katr16}, one should obtain a geometric 
version of the Hodge-elliptic genus. 
In this section, we therefore recall the definitions of the relevant 
geometric genera, according to the proposal of \cite{katr16}.
\smallskip

In the following, let $X$ denote a compact $D$-dimensional complex manifold,
and 
$E\longrightarrow X$  a holomorphic vector bundle on $X$. 
Recall that the \textsc{holomorphic Euler characteristic} of $E$ is given by
$$
\chi(E) = \sum_{j=0}^D (-1)^j \dim H^j(X,E).
$$
Following \cite{katr16}, for $u\in\C^\ast$, we may also introduce
$$
\chi^u(E) := \sum_{j=0}^D (-u)^j \dim H^j(X,E),
$$
which might be  dubbed the \textsc{Hodge-holomorphic Euler characteristic}.
For any formal variable $x$, we use the shorthand notations
$$
\Lambda_x E := \bigoplus_{p=0}^\infty x^p \Lambda^p E,\quad
S_x E:= \bigoplus_{p=0}^\infty x^p S^p E,
$$
where $\Lambda^p E,\, S^p E$ denote the $p^{\rm th}$ exterior and  
symmetric powers of $E$, respectively. 
We are now ready to define the \textsc{complex elliptic genus}
and the \textsc{complex Hodge-elliptic genus}:
\begin{definition}\label{geoell}
Let $X$ denote a compact complex $D$-dimensional  manifold, and
$T:=T^{1,0}X$ its holomorphic tangent bundle. With $q,\,y$ as in
our Assumption \mb{\ref{SCFTass}}, let 
$$
\mathbb E_{q,-y}
:= y^{-{D\over2}}  \bigotimes_{n=1}^\infty \left( \Lambda_{-yq^{n-1}} T^\ast\otimes \Lambda_{-y^{-1}q^n} T
\otimes S_{q^n} T^\ast\otimes S_{q^n} T\right),
$$
viewed as a formal power series with variables $\smash{y^{\pm{1\over2}}},\,q$,
whose coefficients are holomorphic vector bundles on $X$,
$$
\mathbb E_{q,-y} 
= y^{-{D\over2}} \bigoplus_{\ell=0}^\infty \bigoplus_{m=-D}^{D} \mathcal T_{\ell,m} (-y)^m q^\ell.
$$
Then, following \mb{\cite{hi88,wi88,kr90}}, the \textsc{complex elliptic genus} of $X$ is
$$
\EEE(X;\tau, z) := \chi( \mathbb E_{q,-y} )
= y^{-{D\over2}} \sum_{\ell=0}^\infty \sum_{m=-D}^{D} \chi(\mathcal T_{\ell,m}) (-y)^m q^\ell.
$$
With $u:=\exp(2\pi i\nu)$ for 
$\nu\in\C$,  following \mb{\cite{katr16}}, the \textsc{complex Hodge-elliptic genus} 
is\footnote{up to the prefactor $u^{-{D\over2}}$ which presumably is omitted in \cite[\S3]{katr16}
only due to a typo}
$$
\EEE_{\rm Hodge}(X;\tau, z,\nu) 
:= u^{-{D\over2}} \chi^u( \mathbb E_{q,-y} )
= (uy)^{-{D\over2}} \sum_{\ell=0}^\infty \sum_{m=-D}^{D} \chi^u(\mathcal T_{\ell,m}) (-y)^m q^\ell.
$$
\end{definition}
\noindent
From the very definition, it is clear that $\EEE_{\rm Hodge}$ is a natural and 
interesting refinement of the complex elliptic genus:
$$
\EEE_{\rm Hodge}(X;\tau,z,\nu=0) = \EEE(X;\tau,z).
$$
Note  that the holomorphic vector bundles $\mathcal T_{\ell,m}\longrightarrow X$
in the above definition, by construction, are sums of tensor products of exterior
and symmetric powers of the holomorphic tangent bundle $T$ and its dual $T^\ast$.
Nevertheless, one should expect $\EEE_{\rm Hodge}(X;\tau, z,\nu)$ to depend
on the choice of the complex structure on $X$, since the dimensions of the 
cohomology spaces $H^j(X,\mathcal T_{\ell,m})$ may jump, as the complex
structure of $X$ varies (see, e.g., \cite{huy95,bph92,amp12} for some examples
of this phenomenon\footnote{We thank Emanuel Scheidegger for pointing these
references out to us.}). 
For a general Calabi-Yau manifold $X$, we  
therefore always assume
that a fixed complex structure has been chosen.
However, if $X$ is a K3 surface,
then the Bochner principle suffices to prove that 
$\EEE_{\rm Hodge}(X;\tau, z,\nu)$ is independent of this choice, as is shown in
\cite{katr16}. Moreover, for complex tori $X=T^D$, the holomorphic tangent bundle is
trivial, yielding $\EEE_{\rm Hodge}(X;\tau, z,\nu)$ independent of the choice of the complex
structure, as well, as we shall confirm in Sect.~\ref{warm}.
\medskip

In \cite{katr16}, it is claimed that one should expect that 
$\EEE_{\rm Hodge}(X;\tau, z,\nu)$
agrees with the conformal field theoretic Hodge-elliptic
genus of sigma models on $X$ in a large volume limit. 
As we shall see in Prop.~\ref{chiralisnotHodge}, below, this claim does not hold when $X$ is
a K3 surface, unless K3 theories generically
have an extended  chiral algebra beyond the 
$N=4$ superconformal algebra at central charge $c=6$. 
In fact, in Sect.~\ref{comments}, we show that such a
generically extended chiral algebra would contradict the results of Kapustin \cite{ka05}
on large volume limits of topologically half twisted sigma models.
As we shall explain in Sect.~\ref{genericity},  
we do not expect the claim to hold in general, except
for few examples, like complex tori. 
\subsection{Generic genera?}\label{genericity}
In this section, for a compact Calabi-Yau manifold $X$
that obeys certain additional assumptions,
we introduce the notions of {\em generic 
conformal field theoretic Hodge-elliptic genus} and
of  {\em chiral Hodge-elliptic genus}. 
In the string theory literature, the prior would  probably rather be called
an {\em infinite volume limit} of the conformal field theoretic Hodge-elliptic
genus of sigma models on $X$, 
and we expect it to agree with the chiral Hodge-elliptic genus of $X$,
in general.
\smallskip
 
Assume that for some compact $D$-dimensional Calabi-Yau manifold $X$,
we have a notion of a moduli space of SCFTs 
that obey our Assumption \ref{SCFTass} and that allow a 
geometric interpretation by  $X$.  This is the case,
for example, if $X$ is a complex torus or a K3 surface, using 
\cite{cent85,na86,asmo94,nawe00} and Def.~\ref{defk3cft},
where the geometric moduli are naturally expressed in terms
of a hyperk\"ahler structure, the volume and a B-field on $X$. 
In general,
we lack a clean mathematical definition of the appropriate
moduli spaces. The more mathematically inclined reader
is therefore invited to restrict their attention to tori and K3 surfaces,
at least in the context of the generic conformal field theoretic
Hodge-elliptic genus introduced in Def.~\ref{genericCFTHEG}, below.

However, it is largely believed that for 
general Calabi-Yau manifolds $X$,
there is a notion of a moduli space of SCFTs arising as 
non-linear sigma models on $X$. If $X$ is not hyperk\"ahler,
then the moduli space, at least locally, 
is expected to decompose into a product of complex
structure and (complexified) K\"ahler parameter spaces of $X$  \cite{di87,digr88}.  
Mathematically, these two factors exhibit quite distinct behaviours. Since 
both the complex and the chiral Hodge elliptic genus  of $X$
are expected to depend on the complex structure of $X$, in general,
we then assume that $X$ is equipped with a fixed choice of complex 
structure, requiring the {\em\mb{(}complexified\mb{)} K\"ahler moduli} to be generic 
when referring to generic quantities.
If $X$ is hyperk\"ahler, then the situation is different, since 
similarly to the K3 case, the relevant moduli spaces are not even locally
expected to decompose into 
a product of complex structure and (complexified) K\"ahler moduli
spaces. For non-linear sigma models on 
hyperk\"ahler manifolds $X$ we therefore consider {\em all moduli}, 
when we refer to  generic quantities,
so such $X$ is  not assumed to be equipped with a fixed complex 
structure. This distinction is motivated by the 
treatment of infinite volume limits in the string theory literature, particularly 
in \cite{katr16}, and by the fact that for K3 surfaces, both the complex
and the chiral Hodge-elliptic genus are independent of the complex structure,
as we shall see below. In fact, the conformal field theoretic Hodge-elliptic genus 
at generic (complexified) K\"ahler moduli is also independent of all complex
structure moduli for SCFTs that arise as non-linear sigma models
with  some complex torus as target
(see Prop.~\ref{torusgenera}). The same behaviour is predicted 
for K3 theories by string theory,
as we shall see in the proof of Prop.~\ref{assumptionproved}.
It is left for future work to decide whether for higher-dimensional hyperk\"ahler 
manifolds $X$,
a  refinement of our notion of generic conformal field theoretic
Hodge-elliptic genus is more adequate.

We add the assumption that spectral
data depend on the moduli at least continuously, which is known to
hold true for complex tori and orbifolds thereof.
Since in any SCFT that obeys our Assumption \ref{SCFTass}, 
by Def.~\ref{cftellgen},  $\widetilde\H^R$ is the
kernel of the linear operator $\qu L_0-{D\over8}\mb{id}$ on $\H^R$,
the  coefficients of the conformal field theoretic Hodge-elliptic genus
$\EEE_{\rm Hodge}^{\rm CFT}(\H;\tau, z,\nu)$ are the dimensions of common
eigenspaces of the linear operators 
$\qu L_0-{D\over8}\mb{id}$,
$L_0-{D\over8}\mb{id}$,  $J_0$ and $\qu J_0$
on $\H^R$. By our assumptions,  as was explained in the
discussion of Def.~\ref{cftellgen}, all the eigenvalues
of these linear operators on $\widetilde\H^R$
are restricted to values in ${1\over2}\Z$. 
This means that  the {\em generic} dimensions of these 
eigenspaces in $\H^R$ on our moduli space yield the 
{\em maximal lower bounds} of these dimensions\footnote{This behaviour is well known from 
dimension theory in commutative algebra and is sometimes
called {\em upper semicontinuity}, e.g.\  in \cite{katr16}.}.
Let us introduce a generating function for these generic dimensions:
\begin{definition}\label{genericCFTHEG}
Let $X$ denote a compact Calabi-Yau manifold 
of dimension $D$.
Consider the moduli space of SCFTs that obey Assumption \mb{\ref{SCFTass}} 
and that allow a geometric interpretation by $X$. 
Here, we assume that $X$ is equipped with a fixed complex
structure, unless $X$ is hyperk\"ahler. In the hyperk\"ahler case, 
we only fix the diffeomorphism type of $X$.

\noindent
For $h\in{D\over8}+\N$ and $Q,\qu Q\in {D\over2}+\Z$,
and with notations as in Def.~\mb{\ref{cftellgen}}, let
$$
N_{h,Q,\qu Q} := \inf\left( \dim\left\{ \phi\in\widetilde\H^{\rm R} \mid 
L_0\phi = h\phi, \; J_0\phi=Q\phi,  \; \qu J_0\phi=\qu Q\phi  \right\} \right), 
$$
where the infimum is taken over all SCFTs within the moduli space. Then
the \textsc{generic conformal field theoretic
Hodge-elliptic genus} of $X$ is given by
$$
\EEE^0_{\rm Hodge}(X;\tau,z,\nu) 
:= \sum_{h,Q,\qu Q} (-1)^{Q-\qu Q} N_{h,Q,\qu Q} \cdot y^Q u^{\qu Q}q^{h-{D\over8}} ,
$$
where the sum runs over $h\in{D\over8}+\N$ and $Q,\qu Q\in {D\over2}+\Z$.
\end{definition}
By the above definition of the generic conformal field theoretic
Hodge-elliptic genus, for a Calabi-Yau manifold $X$ and for every
SCFT in the moduli space of theories with geometric interpretation by  $X$ 
and with space of states $\H$ as in
Def.~\ref{cftellgen}, there is a space $\widehat\H^R$ which injects
into $\widetilde\H^R$, $\widehat\H^R\hookrightarrow\widetilde\H^R$,
with the following two properties: first,
\begin{equation}\label{genericsubspace}
\EEE^0_{\rm Hodge}(X;\tau,z,\nu)
= \tr\nolimits_{\widehat\H^{\rm R}}\left( (-1)^{J_0-\overline J_0}y^{J_0} u^{\overline J_0}
q^{L_0-{D\over8}} \right),
\end{equation}
where $\widehat\H^{\rm R}$ is a representation
of the left-moving $N=2$ superconformal algebra, extended by $\qu J_0$. 
Second, if the superconformal field theories in our moduli space share a common
extended chiral algebra $\AAA$, then $\widehat\H^{\rm R}$ is a representation
of $\AAA$, extended by $\qu J_0$. 

At fixed values of $h\in{D\over8}+\N$, $Q\in{D\over2}+\Z$, by the defining 
properties of our SCFTs, there are only finitely many values of $\qu Q\in{D\over2}+\Z$
such that $N_{h,Q,\qu Q}\neq0$. Hence generically, the common eigenspaces
of $L_0,\, J_0,\, \qu J_0$ in $\widetilde\H^{\rm R}$ with eigenvalues $h,\,Q,\,\qu Q$
at fixed $h,\, Q$, for all $\qu Q$ have dimension $N_{h,Q,\qu Q}$. 
Since the conformal-field theoretic elliptic genus is invariant
on the moduli space, we therefore have
\begin{equation}\label{genericsubspaceEG}
\EEE^{\rm CFT}(\H;\tau,z) 
= \tr\nolimits_{\widehat\H^{\rm R}} \left( (-1)^{J_0-\qu J_0} y^{J_0} q^{L_0-{D\over8}} \right)
\ee{genericsubspace}
\EEE^0_{\rm Hodge}(X;\tau,z,\nu=0).
\end{equation}
As mentioned above, in string theory\footnote{see, for 
example, \cite[\S3, page 253]{katr16}}, 
$\EEE^0_{\rm Hodge}$ would 
probably be called an {\em infinite volume limit} of the conformal 
field theoretic Hodge-elliptic genus, which for general Calabi-Yau
manifolds refers to a point in an appropriate 
boundary of the moduli space where the dependence on all
(complexified) K\"ahler parameters is lost. For our discussion,
it is more appropriate to focus on {\em generic} values of these
parameters instead of a limit where they are infinite.
For K3 surfaces, by our Definition \ref{genericCFTHEG},
we consider all moduli of K3 theories at generic values. 
Then, the notion of \textsl{infinite volume limit}  seems
a little out of place, and we prefer the notion of \textsl{generic space 
of states}.

It is important to note  that the very restrictions on the spectra of 
$J_0$ and $\qu J_0$ on $\H^R$ in our Assumption \ref{SCFTass}
are crucial in order to expect any meaningful quantity to arise
from Def.~\ref{genericCFTHEG}. 
Indeed, generic dimensions of common eigenspaces
of $\qu L_0,\,L_0,\,J_0$ and $\qu J_0$, for generic choices
of eigenvalues, are zero.
In our view, this is why Kachru's and Tripathy's conformal field 
theoretic Hodge-elliptic genus $\EEE_{\rm Hodge}^{\rm CFT}(\H;\tau, z,\nu)$ 
turns out to be so useful, in contrast to the
partition function.

As was mentioned at the end of Sect.~\ref{geometricellgen},
in Prop.~\ref{chiralisnotHodge}, under one additional assumption which
is commonly believed to hold true, we will disprove the 
expectation of \cite{katr16} that for K3 surfaces $X$,
the generic conformal field theoretic Hodge-elliptic genus
agrees with the complex Hodge-elliptic genus $\EEE_{\rm Hodge}(X;\tau,z,\nu)$
of Def.~\ref{geoell}. 
Indeed, we see no reason for the two quantities to agree, 
in general. Let us now explain why this is so, and offer an alternative 
proposal.
\medskip

In his beautiful work \cite{ka05}, Kapustin has proposed
the following relationship between the non-linear sigma model on
a given compact Calabi-Yau manifold $X$ and  
the (holomorphic) {\em chiral de Rham complex } 
$\Omega_X^{\rm ch}$ of $X$
that was introduced by Malikov, Schechtman, Vaintrob 
\cite{msv98,bo01,boli00,goma03,lili07,bhs08}: Kapustin
 argues that an {\em infinite volume limit}
of the model obtained from the sigma-model by 
a topological half-twist 
\cite{egya90,wi91}  can be identified with
the sheaf cohomology of $\Omega_X^{\rm ch}$. 

Let us recall and explain this in more detail. 
The chiral de Rham complex is a sheaf of vertex operator algebras
which can be defined on any compact complex manifold $X$.
It possesses a bigrading by globally defined 
operators $L_0^{\rm top}$ and $J_0$ and  a 
compatible natural filtration. 
Both these operators descend to the (sheaf) cohomology $H^\ast(X,\Omega_X^{\rm ch})$ of
the chiral de Rham complex, inducing a $\Z_2$-grading by $(-1)^{J_0+j}$ on $H^j(X,\Omega_X^{\rm ch})$.
According to \cite{msv98,bo01,boli00},
the associated graded object for $\Omega_X^{\rm ch}$ is isomorphic to 
the sheaf of sections of 
the virtual bundle $(-y)^{D\over2}\E_{q,y}$ on $X$ that was used in the
Definition \ref{geoell} of the complex elliptic genus of $X$. This implies that 
the latter, given by
the holomorphic Euler characteristic of $\E_{q,-y}$, agrees with the 
graded Euler characteristic of the chiral de Rham complex,
\begin{equation}\label{elliptichiderham}
\EEE(X;\tau,z)
= y^{-{D\over2}} \sum_{j=0}^D (-1)^j \tr\nolimits_{H^j(X,\Omega^{\rm ch}_X)} 
\left( (-1)^{J_0} y^{J_0}  q^{L_0^{\rm top}} \right).
\end{equation}
Note that $H^\ast(X,\Omega_X^{\rm ch})$ carries the action of a
{\em topological}\footnote{that is, topologically twisted, according to \cite{egya90}} $N=2$
superconformal algebra at rank $3D$, according to 
\cite[Prop.~3.7 and Def.~4.1]{bo01}, which 
is extended to an $N=4$ superconformal algebra 
if $X$ is hyperk\"ahler \cite{bhs08,he09} (see also \cite[\S2]{so17}).

Kapustin's interpretation of the cohomology of the chiral de Rham complex
$\Omega_X^{\rm ch}$ as infinite volume limit of the topologically half-twisted
sigma model on $X$ explains why one might expect the complex 
elliptic genus of $X$ to agree with the  conformal
field theoretic elliptic genus of the sigma model. Note that  
$H^\ast(X,\Omega_X^{\rm ch})$ is thus interpreted as
infinite volume limit of the {\em Neveu-Schwarz sector} 
of the sigma model, after the topological half-twist. The latter is
indicated in (\ref{elliptichiderham}) by the
lack of anti-holomorphic contributions. The use of $L_0^{\rm top}=L_0-{1\over2}J_0$,
with $L_0$ the untwisted Virasoro zero-mode, accounts for the fact
that $H^\ast(X,\Omega_X^{\rm ch})$ naturally carries the action of a 
{\em topologically twisted} $N=2$ superconformal algebra.
To obtain the elliptic genus, one needs to perform a spectral flow
from the Neveu-Schwarz to the Ramond sector. This is reflected
in (\ref{elliptichiderham})  by the fact that the trace of
$(-1)^{J_0}y^{J_0}q^{L_0^{\rm top}}
=(-1)^{J_0}q^{D\over8}(yq^{-{1\over2}})^{J_0}q^{L_0-{D\over8}}$
is taken, instead of $(-1)^{J_0}y^{J_0}q^{L_0-{D\over8}}$.

This reasoning also implies that
$H^j(X,\E_{q,-y})$ cannot be expected to 
be isomorphic to (a graded object of) 
$H^j(X,\Omega_X^{\rm ch})$. In fact, by the above,
$H^j(X,\E_{q,-y})$ arises on the first sheet of the spectral sequence
that is obtained from our filtered complex, while $H^j(X,\Omega_X^{\rm ch})$
requires the limit of that spectral sequence.
Therefore, we do not 
expect the respective Hodge-elliptic genera
to agree, not even in an infinite volume limit. As an 
alternative, we propose to define a Hodge-elliptic genus
using the (holomorphic)
chiral de Rham complex. To fully appreciate such a definition, 
a few more details about this sheaf are helpful.
First note that in \cite{ka05}, a {\em Dolbeault resolution} of
the chiral de Rham complex is introduced, which essentially extends
$\Omega_X^{\rm ch}$ by  
additional {\em anti-holomorphic fields}. 
All these additional fields are assumed to be constant. 
As Grimm explains in his thesis \cite[\S5.1]{gr16},
this restricts one
to the realm of real analytic differential forms, where one lacks a partition of unity.
The relevant sheaves, therefore, are {\em not fine}. Also inspired 
by the works of  Lian and Linshaw \cite{lili07},
Grimm instead constructs a resolution that sits in between the two that
are suggested in \cite{ka05,lili07}, respectively. He proves that the
resulting sheaves yield an acyclic resolution, thus allowing him access to 
explicit calculations of chiral de Rham cohomology,
in some examples. The important point, for
our purposes, is the existence of a well-defined operator 
$\qu J_0$ on the sections of the chiral de Rham complex, 
which descends to cohomology, as follows from Grimm's construction,
with $H^j(X,\Omega_X^{\rm ch})$ arising as 
the kernel of $\qu J_0-j\,\mb{id}$ 
in $H^\ast(X,\Omega_X^{\rm ch})$ by \cite[Thm.~5.1.7]{gr16}.
As was mentioned before, this is already implicit in the constructions
of \cite{bo01,boli00}, where the $\Z_2$-grading by what we now recognize
as\footnote{$(-1)^{J_0-\qu J_0}=(-1)^{J_0+\qu J_0}$ since by Assumption
\ref{SCFTass}, all eigenvalues of
$\qu J_0$ in the Neveu-Schwarz sector are  integral.} 
$(-1)^{J_0-\qu J_0}$ is introduced.
This motivates the following
\begin{definition}\label{chiralHEG}
Let $X$ denote a compact Calabi-Yau manifold of dimension $D$,
and $\Omega_X^{\rm ch}$ its \mb{(}holomorphic\mb{)}
chiral de Rham complex. The 
\textsc{chiral Hodge-elliptic genus} of $X$ is defined by
\begin{eqnarray*}
\EEE^{\rm ch}_{\rm Hodge}(X;\tau,z,\nu)
&:=& (yu)^{-{D\over2}} \tr\nolimits_{H^\ast(X,\Omega^{\rm ch}_X)} 
\left( (-1)^{J_0-\qu J_0} y^{J_0} u^{\qu J_0} q^{L_0^{\rm top}} \right)\\
&=& (yu)^{-{D\over2}} \sum_{j=0}^D (-u)^j
\tr\nolimits_{H^j(X,\Omega^{\rm ch}_X)} 
\left( (-1)^{J_0} y^{J_0}  q^{L_0^{\rm top}} \right).
\end{eqnarray*}
\end{definition}
Let us discuss some of the properties of this new Hodge-elliptic genus.
The chiral Hodge-elliptic genus is a natural refinement of
the complex elliptic genus, as is immediate from the above
Def.~\ref{chiralHEG} along with (\ref{elliptichiderham}):
\begin{equation}\label{chHEG0EG}
\EEE^{\rm ch}_{\rm Hodge}(X;\tau,z,\nu=0)=\EEE(X;\tau,z).
\end{equation}
Following Kapustin, $H^\ast(X,\Omega^{\rm ch}_X)$
is interpreted as infinite volume limit of 
the Neveu-Schwarz sector of a topologically half-twisted sigma model
on $X$, which as advertised above should be 
a space that can be injectively mapped into the space of states 
of any sigma model on $X$. In this sense, $H^\ast(X,\Omega^{\rm ch}_X)$
is a common subspace of all such theories.
Spectral flow maps this space to 
$\widehat\H^{\rm R}\hookrightarrow\widetilde\H^{\rm R}$
in (\ref{genericsubspace}), where the restriction
to $\ker\left(\qu L_0-{\qu c\over24}\cdot\mb{id}\right)$
with $\qu c=3D$
in the Definition \ref{cftellgen} of $\widetilde\H^{\rm R}$
implements the topological half-twist, as mentioned above. 
This explains why we expect
the {\em chiral Hodge-elliptic genus} of Def.~\ref{genericischiral},
rather than the complex Hodge-elliptic genus of \cite{katr16},
to agree with the generic conformal field theoretic 
Hodge-elliptic genus of $X$:
\begin{conjecture}\label{genericischiral}
Consider a compact Calabi-Yau manifold $X$. The
generic conformal field theoretic Hodge-elliptic genus 
of SCFTs with geometric interpretation by $X$
\mb{(}Def.~\mb{\ref{genericCFTHEG})}
agrees with the chiral Hodge-elliptic genus of $X$
\mb{(}Def.~\mb{\ref{chiralHEG})},
$$
\EEE^0_{\rm Hodge}(X;\tau,z,\nu) = \EEE^{\rm ch}_{\rm Hodge}(X;\tau,z,\nu).
$$
\end{conjecture}
In Sect.~\ref{warm}, we will see that it is straightforward to 
prove this conjecture if $X$ is a complex torus, see 
Prop.~\ref{torusgenera}. Moreover,
Prop.~\ref{proposedHEG} shows 
that for K3 surfaces $X$, this conjecture holds under 
the natural assumption that the generic
chiral algebra of all K3 theories is precisely the $N=4$
superconformal vertex operator algebra at central 
charge $c=6$. On the other hand, 
Prop.~\ref{chiralisnotHodge} states that under the
same assumption, the complex Hodge-elliptic genus of K3 {\em differs}
from the generic conformal field theoretic Hodge-elliptic
genus of K3 theories. Finally, Prop.~\ref{assumptionproved}
shows that our assumption on the generic chiral algebra of K3
theories holds true if the (sheaf) cohomology
of the (holomorphic) chiral de Rham complex of a K3 surface $X$
can indeed be identified with the infinite volume limit
of topologically half-twisted sigma models on $X$ in the sense
explained above, and as argued by Kapustin \cite{ka05}.
\smallskip

There is no reason to expect the chiral Hodge-elliptic
genus to be independent of the choice of complex structure 
on $X$. However, in Props.~\ref{torusgenera} and \ref{proposedHEG}
we will see that $\EEE_{\rm Hodge}^{\rm ch}(X;\tau,z,\nu)$ is independent
of that choice if $X$ is a complex torus or a K3 surface. In this respect,
the chiral Hodge-elliptic genus behaves analogously to the complex
Hodge-elliptic
genus of \cite{katr16}.

Finally
recall that for superconformal field theories which obey our Assumption \ref{SCFTass},
on the level of the $N=(2,2)$ superconformal algebra
with generating fields  $(T,J,G^+,G^-; \qu T,\qu J,\qu G^+,\qu G^-)$,
{\em mirror symmetry}
acts as an outer automorphism induced by
\begin{eqnarray*}
(T,J,G^+,G^-; \qu T,\qu J,\qu G^+,\qu G^-)&\longmapsto& (T,-J,G^-,G^+; \qu T,\qu J,\qu G^+,\qu G^-)\\
\quad\mb{ or }\\
(T,J,G^+,G^-; \qu T,\qu J,\qu G^+,\qu G^-)&\longmapsto& (T,J,G^+,G^-; \qu T,-\qu J,\qu G^-,\qu G^+),
\end{eqnarray*}
where the choice between the two should solely amounts to a choice of normalization
\cite{lvw89}.  Assumption \ref{SCFTass} includes the requirement that
in the Ramond sector, all eigenvalues of $J_0,\,\qu J_0$ belong to ${D\over2}+\Z$.
Therefore,
if Conjecture \ref{genericischiral} holds, then it  in particular implies
\begin{conjecture}\label{mirror}
Let $(X,\widetilde X)$ denote a mirror pair of compact Calabi-Yau manifolds
of complex dimension $D$. 
Then the corresponding chiral Hodge-elliptic genera obey 
$$
\EEE_{\rm Hodge}^{\rm ch}(\widetilde X;\tau, z,\nu)
=(-1)^D \EEE_{\rm Hodge}^{\rm ch}(X;\tau, -z,\nu)
=(-1)^D \EEE_{\rm Hodge}^{\rm ch}(X;\tau, z,-\nu).
$$
\end{conjecture}
This is an immediate generalization of the corresponding behaviour of
the elliptic genera for mirror pairs of Calabi-Yau manifolds, which in 
turn is well-established, see e.g.\ \cite[\S1($\ast$)]{boli00}.
\section{Warmup: Hodge-elliptic genera for tori}\label{warm}
As a warm up, in this section we  discuss the complex elliptic genus
and the Hodge-elliptic genera for complex tori $T^D$ of any dimension $D$,
all of which can be calculated explicitly with little effort.
Although for the experts,  this is an easy exercise, we find it useful
to include its solution in this note, since  for the Hodge-elliptic genera introduced in the 
previous section it  may also explain some of
the misconceptions in the literature.
\smallskip

First, since the holomorphic tangent bundle of $T^D$ is trivial,
by classical index theory, the complex elliptic genus of $T^D$ is
\begin{equation}\label{toruscpxellg}
\EEE(T^D;\tau,z) = 0.
\end{equation}
The complex Hodge-elliptic genus
$\EEE_{\rm Hodge}(T^D;\tau,z)$ can also be explicitly
calculated
(c.f.\  \cite[\S3]{katr16}, whose normalization of $\vartheta_1(\tau,z)$ differs by
a prefactor of $-i$ from ours),
\begin{equation}\label{toruskatr}
\EEE_{\rm Hodge}(T^D;\tau,z) = \left( {-i\vartheta_1(\tau,z)\over\eta(\tau)^3} 
\cdot(u^{-{1\over2}}-u^{{1\over2}})\right)^D.
\end{equation}
Furthermore, the sheaf cohomology of the (holomorphic)
chiral de Rham complex
for $T^D$ consists of the classical Dolbeault cohomology of the torus,
tensorized by polynomials in $D$ copies of
the modes $(b_m,\, a_m,\, \Phi_m, \Psi_m)_{m>0}$ of a bc-$\beta\gamma$ 
system
\cite[Thm.~5.2.1]{gr16}. On the classical Dolbeault cohomology
$H^{j,k}(T^D,\C)\hookrightarrow H^k(T^D,\Omega_{T^D}^{\rm ch})$, 
one has $L_0^{\rm top}\equiv 0,\; J_0\equiv j\cdot\mb{id},\;
\qu J_0\equiv k\cdot\mb{id}$, and thus
\begin{align}\label{torusHodgechiral}
\EEE_{\rm Hodge}^{\rm ch}(T^D;&\tau,z,\nu) \nonumber\\
&\eet{Def.~\ref{chiralHEG}}\; 
(yu)^{-{D\over2}} \sum_{j,k=0}^D (-1)^{j-k} {D\choose j}{D\choose k} y^j u^k
\cdot \prod_{m=1}^\infty \left( {(1-yq^m)(1-y^{-1}q^m)\over(1-q^m)^2} \right)^D \nonumber\\
&\eet{App.~\ref{N4characters}}\; 
\left( {-i\vartheta_1(\tau,z)\over\eta(\tau)^3} 
\cdot(u^{-{1\over2}}-u^{{1\over2}})\right)^D\\
&\ee{toruskatr}\;\;\;\;
\EEE_{\rm Hodge}(T^D;\tau,z,\nu).\nonumber
\end{align}
Note that this result is compatible with Conjecture \ref{mirror}, since the mirror
of a complex $D$-torus $T^D$ is a complex $D$-torus.
\medskip

To compare to the results obtained from the conformal field theoretic Hodge-elliptic
genus, recall that any non-linear sigma model with target $T^D$ and space
of states $\H$ has a partition function whose $\widetilde{\rm R}$-sector 
(c.f.\ Assumption \ref{SCFTass})
takes the following form:
\begin{equation}\label{torusparti}
Z_{\widetilde{\rm R}}^{T^D}(\tau,z) 
= \displaystyle Z_\Gamma(\tau)  \left| {\vartheta_1(\tau,z)\over\eta(\tau)} \right|^{2D},
\quad\mb{ where }\quad
Z_\Gamma(\tau)
=\sum_{\gamma=(\gamma_L,\gamma_R)\in\Gamma} 
{q^{{1\over2}\gamma_L\cdot\gamma_L}\overline q^{{1\over2}\gamma_R\cdot\gamma_R}\over
|\eta(\tau)|^{4D}}
\end{equation}
depends on the moduli of the theory via the charge lattice 
$\Gamma\subset\R^{D,D}$
(see \cite[\S2.1]{we15} for a recent review adapted to our purposes).
Here and in the following, $\R^{D,D}=\R^D\oplus\R^D$ is equipped with
the scalar product
$$
\forall ({Q},\qu{Q}),\, ({Q}^\prime,\qu{Q}^\prime)\in\R^{D,D}\colon\qquad
({Q},\qu{Q}) \bullet ({Q}^\prime,\qu{Q}^\prime)
:= {Q}\cdot {Q}^\prime - \qu{Q}\cdot \qu{Q}^\prime,
$$
where  ${Q}\cdot{Q}^\prime\in\R$ denotes the standard scalar product of
${Q},\,{Q}^\prime\in\R^D$.
According to Def.~\ref{cftellgen}, we thus have
\begin{equation}\label{torustrue}
\EEE_{\rm Hodge}^{\rm CFT}(\H;\tau,z,\nu)
= \sum_{\stackrel{\gamma=(\gamma_L,\gamma_R)\in\Gamma}{\mb{\tiny with } \gamma_R=0}} 
{q^{{1\over2}\gamma_L\cdot\gamma_L}\over \eta(\tau)^{2D}}\cdot
\left( {-i\vartheta_1(\tau,z)\over\eta(\tau)} \cdot(u^{-{1\over2}}-u^{{1\over2}})\right)^D,
\end{equation}
hence 
\begin{equation}\label{CFTHEG0}
\EEE^{\rm CFT}(\H;\tau,z) 
\ee{CFTHEGisCFTEG}
\EEE_{\rm Hodge}^{\rm CFT}(\H;\tau,z,\nu=0) \ee{torustrue} 0
\ee{toruscpxellg} \EEE(T^D;\tau,z),
\end{equation}
as expected by what was said in Sect.~\ref{setup}, and as in fact
is well-known. Moreover, for any given value of $h\in\R$ with $h>0$,
and for almost all possible charge lattices $\Gamma$,
one finds
$$
\left\{ \gamma=(\gamma_L,\gamma_R)\in\Gamma
\mid{\textstyle{1\over2}}\gamma_L\cdot\gamma_L=h,\; \gamma_R=0\right\} 
=\emptyset.
$$
In fact, there is a dense subset of the moduli space of non-linear
sigma models with target 
$T^D$ where the charge lattices $\Gamma$ 
obey
\begin{equation}\label{noenhancement}
\left\{ \gamma=(\gamma_L,\gamma_R)\in\Gamma\mid\gamma_R=0\right\} 
=\{ 0 \}.
\end{equation}
The subset of the moduli space where the charge lattice $\Gamma$ disobeys condition
(\ref{noenhancement}) is the union of those subsets where $\Gamma$ contains
some vector $(\gamma_L,\gamma_R)$ with $\gamma_R=0$ and
$\gamma_L\cdot\gamma_L=2N$, $N\in\N$. This is a countable union of nowhere
dense sets.  Hence  the theories 
in the dense subset of the moduli space where (\ref{noenhancement}) holds
are the {\em generic} ones\footnote{To see that this is true already when 
in any geometric
interpretation, one only varies the volume parameter, using 
\cite[(1.11)]{nawe00} as standard convention for charge vectors, note
that for $T^D\cong\R^{2D}/\Lambda$, every
$\gamma_R$ has the form 
$\gamma_R={1\over\sqrt2}(\mu-B\lambda-\lambda)$ with 
$\lambda\in\Lambda$, $\mu\in\Lambda^\ast=\left\{ y\in\R^{2D} \mid
y\cdot\lambda\in\Z \;\forall\lambda\in\Lambda\right\}$ and a 
skew-symmetric $B\in\mb{End}(\R^{2D})$. Varying the volume parameter
amounts to scaling $\Lambda$ by some $t\in\R$, where 
$(t\Lambda)^\ast=t^{-1}\Lambda^\ast$.}.
For the generic conformal field theoretic Hodge-elliptic
genus of Def.~\ref{genericCFTHEG}, (\ref{torustrue})
thus immediately implies
\begin{equation}\label{torusgenCFTHEG}
\EEE_{\rm Hodge}^0(T^D;\tau,z,\nu)
=
\left( {-i\vartheta_1(\tau,z)\over\eta(\tau)^3} \cdot(u^{-{1\over2}}-u^{{1\over2}})\right)^D.
\end{equation}
This may be used to confirm the ideas presented in Sect.~\ref{genericity},
in the case of complex tori. First,
$$
\EEE_{\rm Hodge}^0(T^D;\tau,z,\nu) 
\eet{(\ref{torusgenCFTHEG}),\,(\ref{torusHodgechiral})}
\EEE_{\rm Hodge}^{\rm ch}(T^D;\tau,z,\nu),
$$
that is, Conjecture \ref{genericischiral} holds for complex tori.
Now let us determine a candidate for the {\em generic space} $\widehat\H^{\rm R}$
that was described in Sect.~\ref{genericity}.
By construction, SCFTs with target
$T^D$ share a common chiral algebra that can be described
as $N=2D$ superextension ${\mathcal A}$ of a
$\mathfrak{u}(1)^{2D}_L$ current algebra.
In general, any irreducible representation
$\H_h$ of $\mathcal A$ in the Ramond
sector has lowest weight 
$h+{D\over8}$  with $h\geq0$, and one finds 
\begin{equation}\label{toruscharacter}
\mbox{Tr}_{\H_h}\left( (-1)^{J_0}y^{J_0} q^{L_0-{D\over8}}  \right)
= q^h \left( {-i\vartheta_1(\tau,z)\over\eta(\tau)^3} )\right)^D
= q^h \left( (y^{-{1\over2}}-y^{{1\over2}})^D + {\mathcal O}(q) \right).
\end{equation}
In particular, if $\phi\in\H_h$ with $L_0\phi={D\over8}\phi$, then $\phi\in\H_0$,
and analogously for right-movers. 
Moreover, for a non-linear sigma model with target $T^D$ and charge lattice
$\Gamma\subset\R^{D,D}$, the Ramond sector has the form
$$
\H^R 
= \bigoplus_{\gamma=(\gamma_L,\gamma_R)\in\Gamma} 
\H_{{1\over2}\gamma_L\cdot\gamma_L}\otimes \qu{\H}_{{1\over2}\gamma_R\cdot\gamma_R}.
$$
With $\widetilde{\HHH}:=\left\{ \phi\in\qu\H_0 \mid \qu L_0\phi={D\over8}\phi\right\}$,
viewed as a representation of the Lie algebra of type
$\mathfrak{u}(1)$ generated by $\qu J_0$,
according to Def.~\ref{cftellgen} and the discussion in Sect.~\ref{genericity}, we thus have
$$
\widetilde\H^R = \bigoplus_{\stackrel{\gamma=(\gamma_L,\gamma_R)\in\Gamma}{\mb{\tiny with } \gamma_R=0}} 
\H_{{1\over2}\gamma_L\cdot\gamma_L}\otimes \widetilde{\HHH}.
$$
Therefore, by what was said above about the generic behaviour
of the charge lattices $\Gamma$ on the moduli space,
$$
\widehat\H^R := \H_0\otimes \widetilde{\HHH}\subset\widetilde\H^R
$$
is a representation of the common chiral algebra $\AAA$, extended by
$\qu J_0$, as required in Sect.~\ref{genericity}. 
The space $\widehat\H^R$ also obeys the other requirement  (\ref{genericsubspace}),
$$
\tr\nolimits_{\widehat\H^{\rm R}}\left( (-1)^{J_0-\overline J_0}y^{J_0} u^{\overline J_0}
q^{L_0-{D\over8}} \right)
\ee{toruscharacter}
\left( {-i\vartheta_1(\tau,z)\over\eta(\tau)^3} \cdot(u^{-{1\over2}}-u^{{1\over2}})\right)^D
\ee{torusgenCFTHEG}
\EEE_{\rm Hodge}^0(T^D;\tau,z,\nu).
$$
We summarize the
above results in 
\begin{proposition}\label{torusgenera}
The {\em conformal field theoretic Hodge-elliptic genus} of a 
non-linear sigma model with target a complex torus $T^D$ 
depends severely on the moduli of the theory, as is readily seen
from \mb{(\ref{torustrue})}. 
The {\em generic conformal field theoretic Hodge-elliptic genus}
in this case agrees with the {\em chiral Hodge-elliptic genus},
confirming Conjecture \mb{\ref{genericischiral}} for complex
tori. In particular, 
there is a dense subset of the moduli space of non-linear
sigma models with target $T^D$ where 
$$
\EEE_{\rm Hodge}^{\rm CFT}(\H;\tau,z,\nu) = \EEE_{\rm Hodge}^{\rm ch}(T^D;\tau,z,\nu).
$$
In fact, the above equation  holds  generically, within the moduli space
of SCFTs that arise as non-linear sigma models with target $T^D$. 

Finally, for complex tori, the {\em chiral Hodge-elliptic genus}
agrees with the {\em complex Hodge-elliptic genus} of \mb{\cite{katr16}}, 
and thus the latter agrees with the {\em generic conformal field theoretic
Hodge elliptic-genus} in this case.
\end{proposition}
Though the final statement of the above proposition is
not claimed explicitly in \cite{katr16}, it is surely known to
the authors, as their discussions in \cite{katr17} indicate.
It is probably the origin of the false expectation that the complex
Hodge-elliptic genus should {\em always} agree with the generic 
conformal field theoretic Hodge-elliptic genus.
Note also that 
the subset of the moduli space of SCFTs with target
$T^D$ where the charge lattice $\Gamma$ does {\em not} obey
(\ref{noenhancement}) contains the set of {\em rational} toroidal SCFTs
and thus is dense.
There, the conformal field theoretic Hodge-elliptic genus differs 
from the complex Hodge-elliptic genus of $T^D$. This is the basis of the 
discussions in \cite{katr17}.
\section{Hodge-elliptic genera for K3}\label{proposal}
This section provides results on the  Hodge-elliptic genera introduced in 
Sect.~\ref{setup}, in the case of K3 theories and K3 surfaces. 
We establish a formula for each of these; 
for the generic conformal field theoretic Hodge-elliptic genus, the result is obtained
under the assumption that generically, the chiral
algebra of a K3 theory is precisely the $N=4$ superconformal
algebra at central charge $c=6$. 
Apart from the fact that in string theory, this assumption is commonly believed
to hold true, Sect.~\ref{comments} provides further evidence
in favour of this assumption, see Prop.~\ref{assumptionproved}.  
In the final section \ref{Mathieu}, we discuss the consequences of our
findings for {\em Mathieu Moonshine}.
\subsection{The conformal field theoretic Hodge-elliptic genus of
Kummer K3 theories}\label{Kummer}
The {\em Kummer construction} is a classical construction of certain K3 
surfaces as $\Z_2$-orbifolds of complex two-tori (see e.g.\ \cite[\S1.3]{we15} for a 
recent summary adjusted to our purposes). Its conformal field theoretic 
counterpart yields examples of K3 theories that are obtained 
by the standard $\Z_2$-orbifold
construction
from non-linear
sigma models with target a complex two-torus  \cite{eoty89}.
In the current section, we  investigate the conformal field 
theoretic Hodge-elliptic genera
in this particular setting. For the experts, this again is an easy exercise;
nevertheless, as we shall see,
there seem to be some misconceptions 
about this in the literature.
\smallskip

Consider a SCFT with space of states
$\H_{\rm orbifold}$,  obtained 
as a standard
 $\Z_2$-orbifold SCFT\footnote{Note that odd $D$ yields $\Z_2$-orbifolds
 that do not allow a geometric interpretation on a compact Calabi-Yau 
 target. Indeed, resolving the singularities of $T^D/\Z_2$ does not yield a Calabi-Yau
 manifold,
 since the holomorphic volume form on $T^D$ is not preserved
 by the $\Z_2$-action. The resulting conformal field theories are well-defined
 nevertheless, with the formulae for the partition function and its 
 $\wt R$-sector as stated.}  from a non-linear sigma model with
target a complex torus $T^D$ 
of dimension $D$ with charge lattice $\Gamma$
as in (\ref{torusparti}).
Using standard conformal field theory techniques
(reviewed, for example, in \cite[\S2.2]{we15}, in a form adjusted to our purposes),
for the $\wt{\rm R}$-sector of the partition function (see Assumption \ref{SCFTass}),
one obtains
%
$$
Z^{\rm orbifold}_{\wt{\rm R}}(\tau,z)= {\textstyle {1\over2} }
\displaystyle\left( Z_\Gamma(\tau)  \left| {\vartheta_1(\tau,z)\over\eta(\tau)} \right|^{2D}
+ \left| {2\vartheta_2(\tau,z)\over\vartheta_2(\tau,0)} \right|^{2D}
+ \left| {2\vartheta_3(\tau,z)\over\vartheta_3(\tau,0)} \right|^{2D}
+ \left| {2\vartheta_4(\tau,z)\over\vartheta_4(\tau,0)} \right|^{2D} \right).\nonumber
$$
This yields 
\begin{align}
\EEE_{\rm Hodge}^{\rm CFT}(\H_{\rm orbifold};\tau,z,\nu)
&= {\textstyle {1\over2} }
\displaystyle\left(
\smash{\sum_{\stackrel{\gamma=(\gamma_L,\gamma_R)\in\Gamma}{\mb{\tiny with } \gamma_R=0}} }
{q^{{1\over2}\gamma_L\cdot\gamma_L}\over \eta(\tau)^{2D}}\cdot
\left( {-i\vartheta_1(\tau,z)\over\eta(\tau)} \cdot(u^{-{1\over2}}-u^{{1\over2}})\right)^D\right.\nonumber\\[10pt]
&\quad\quad\left.
+ \left( {2\vartheta_2(\tau,z)\over\vartheta_2(\tau,0)} 
\cdot(u^{-{1\over2}}+u^{{1\over2}})\right)^{D}
+ \left( {4\vartheta_3(\tau,z)\over\vartheta_3(\tau,0)} \right)^{D}
+ \left( {4\vartheta_4(\tau,z)\over\vartheta_4(\tau,0)} \right)^{D} \right).\nonumber
\end{align}
For $D=2$, inserting $\nu=0$ correctly yields the elliptic genus
$\EEE(X;\tau,z)$ of a K3 surface $X$ as in (\ref{K3ellgen}).
Analogously to the derivation of Prop.~\ref{torusgenera}, we arrive at
\begin{proposition}\label{HEGforZ2orbs}
There is a dense subset of the moduli space of $\Z_2$-orbifolds of
non-linear sigma models with target $T^D$ where the charge lattices 
$\Gamma$ obey \mb{(\ref{noenhancement})}. The corresponding
SCFTs have conformal field
theoretic Hodge-elliptic genus 
\begin{align}
\EEE_{\rm Hodge}^{\rm CFT}(\H_{\rm orbifold};\tau,z,\nu)
&= {\textstyle {1\over2} }
\displaystyle\left(
\left( {-i\vartheta_1(\tau,z)\over\eta(\tau)^3} \cdot(u^{-{1\over2}}-u^{{1\over2}})\right)^D\right.\nonumber\\[10pt]
&\quad\quad\left.
+ \left( {2\vartheta_2(\tau,z)\over\vartheta_2(\tau,0)} 
\cdot(u^{-{1\over2}}+u^{{1\over2}})\right)^{D}
+ \left( {4\vartheta_3(\tau,z)\over\vartheta_3(\tau,0)} \right)^{D}
+ \left( {4\vartheta_4(\tau,z)\over\vartheta_4(\tau,0)} \right)^{D} \right).\nonumber
\end{align}
In fact, the above equation  holds  generically for
$\Z_2$-orbifolds of non-linear sigma models with target 
$T^D$.
\end{proposition}
For $D=2$, the formula proved in Prop.~\ref{HEGforZ2orbs} 
agrees
with the one stated in \cite[\S3, p.~253]{katr16}, where however it is
claimed that this formula yields the conformal field theoretic
Hodge-elliptic genus for the  $\Z_2$-orbifold conformal field
theory obtained from the toroidal theory with target the standard torus
$\R^4/\Z^4$. This claim of \cite[\S3, p.~253]{katr16} is a little misleading,
since it can only
hold true if a non-rational B-field is chosen, 
as to ensure condition (\ref{noenhancement}), i.e.\ 
if the $\Z_2$-orbifold conformal field 
theory under inspection is {\em generic}. The choice of the standard
torus is irrelevant.
\subsection{The generic conformal field theoretic Hodge-elliptic genus of K3 theories}\label{genericK3CFT}
The current section establishes a formula for the generic conformal field
theoretic Hodge-elliptic genus of K3 theories, under the following
\begin{assumptions}\label{K3chiralalgebra} 
The chiral algebra of a generic K3 theory, according 
to Def.~\mb{\ref{defk3cft}}, agrees with the $N=4$ superconformal
algebra at central charge $c=6$. 
\end{assumptions}
This assumption is widely believed to hold true, in  string theory,
though probably the only solid piece of evidence in favour
of this assumption is the lack of a better candidate for such
a generic chiral algebra.
Below, and in Sect.~\ref{comments}, we will present further evidence in favour of 
Assumption \ref{K3chiralalgebra}, see Prop.~\ref{assumptionproved}. The assumption allows us to 
determine the representation $\widehat\H^{\rm R}$ used
in (\ref{genericsubspace}),  resulting in an explicit formula
for the generic conformal field theoretic Hodge-elliptic genus of K3
theories\footnote{A large 
part of the following argument already occurs in \cite[\S4.1]{we14} and thus is
based on \cite{diss} as well as the ideas of \cite{eoty89}.}.
To explain this, in the following, a priori only assume that $\H$ is
the space of states  of a fixed K3 theory
as in Def.~\ref{defk3cft}, and let $X$ denote a K3 surface. 
\medskip

As mentioned in the discussion of Def.~\ref{defk3cft}, 
all K3 theories enjoy $N=(4,4)$ supersymmetry. Hence the
subspace  $\widehat\H^{\rm R}$ of $\H^{\rm R}$ is  
a representation of the left-moving 
$N=4$ superconformal algebra, extended by $\qu J_0$. Moreover,
by our assumptions on SCFTs in general and K3 theories in particular,
$\H^{\rm R}$ decomposes into a direct sum of tensor products of irreducible
representations of the left- and the right-moving $N=4$ superconformal
algebra. Such representations have been classified, and their
characters have been determined in \cite{egta87,egta88a,egta88,egta88b,ta89}. 

There are three types of irreducible
unitary representations of the $N=4$ superconformal algebra at central charge $c=6$, called
the \textsc{vacuum representation},
the \textsc{massless matter representation},
and finally the \textsc{massive matter representations}. The latter form a one-parameter
family indexed by $h\in\mathbb R_{>0}$. For our purposes,  we may
focus on the Ramond sector.
Alluding to
the properties of the corresponding representations
in the Neveu-Schwarz sector $\mathbb H^{\rm NS}$, which are related to the representations
in $\mathbb H^{\rm R}$ by spectral flow,
we denote the respective irreducible unitary representations by
$\mathcal H_0$, $\mathcal H_{\rm mm}$, $\mathcal H_h$ ($h\in\mathbb R_{>0}$).
Indeed, the ground state of the 
vacuum representation in the NS-sector 
is the vacuum of the theory. 

Using ${c\over24}={1\over4}$, the respective characters of the irreducible unitary 
representations of the $N=4$ superconformal algebra at central charge $c=6$
are denoted  by
$$
\chi_a (\tau,z) := 
\mbox{Tr}_{\mathcal H_a} \left( (-1)^{J_0} y^{J_0} q^{L_0-{1\over4}} \right),\;\;
a\in\mathbb R_{\geq0}\cup\{{\rm mm}\}.
$$
For the reader's convenience, we state the explicit formulae
of  \cite{egta88a} for 
these functions in Appendix \ref{N4characters}.
But at this point, we only need the following properties:
\begin{equation}\label{wittenindices}
\begin{array}{rclrcl}
\chi_0(\tau, z=0) &=& -2,  &\chi_{\rm mm}(\tau, z=0) &=& 1, \quad \\[5pt]
\forall h>0\colon\qquad\chi_h(\tau,z) &=& q^h \widehat\chi(\tau,z)&
\mbox{ with }\quad
\widehat\chi(\tau,z) &=&\chi_0(\tau, z) + 2\chi_{\rm mm}(\tau, z),\\[5pt]
&&\mbox{hence }\\
&&&\chi_h(\tau, z=0) &=&  \widehat\chi(\tau, z=0)\;=\;0.
\end{array} 
\end{equation}
The constant $\chi_a (\tau,z=0)$ is called the \textsc{Witten index}  \cite{wi82,wi87,wi88}
of the respective representation.

The transformation properties of the above characters
under the action of the modular group in general are \textsl{not} modular, 
in contrast to the situation at lower supersymmetry, where 
an infinite class of
characters of irreducible unitary representations 
does enjoy modularity. Instead, the massless $N=4$ characters 
exhibit a so-called \textsc{Mock modular} behaviour, as was
already observed -- necessarily using different terminology, then --
in \cite{egta87,egta88a,egta88,egta88b,ta89}. 
\medskip

We may now make an ansatz for a decomposition of $\mathbb H^{\rm R}$ 
into irreducible representations of
the two commuting $N=4$ superconformal algebras, 
$$
\mathbb H^{\rm R} = \bigoplus_{a,\,\overline a\in\mathbb R_{\geq0}\cup\{{\rm mm}\}} 
m_{a,\overline a} \mathcal H_a\otimes \overline{\mathcal H_{\overline a}},
$$
with appropriate non-negative integers $m_{a,\overline a}$. Then the
$\wt{\rm R}$-sector of the partition function of our theory (see Assumption
\ref{SCFTass}) reads 
$$
Z_{\wt{\rm R}}(\tau,z)
= \sum_{a,\,\overline a\in\mathbb R_{\geq0}\cup\{{\rm mm}\}} 
m_{a,\overline a}\cdot \chi_a(\tau,z)\cdot \overline{\chi_{\overline a}(\tau, z)},
$$
which together with Def.~\ref{cftellgen} yields the conformal field theoretic elliptic genus of our CFT as
\begin{equation}\label{ellgenform}
\EEE^{\rm CFT}(\H;\tau,z)  = \sum_{a,\,\overline a\in\mathbb R_{\geq0}\cup\{{\rm mm}\}} 
m_{a,\overline a}\cdot \chi_a(\tau,z)\cdot \overline{\chi_{\overline a}(\tau, z=0)}.
\end{equation}
Inserting  (\ref{wittenindices}) as well as a number of 
known properties of K3 theories (see \cite[\S4.1]{we14} for details), 
this ansatz simplifies to 
\begin{equation}\label{refinedansatz}
\begin{array}{rcl}
\mathbb H^R = \mathcal H_0\otimes \overline{\mathcal H_0}
&\oplus& \displaystyle
20 \mathcal H_{\rm mm}\otimes \overline{\mathcal H_{\rm mm}}
\oplus \bigoplus_{h,\, \overline h\in\mathbb R_{>0} } k_{h,\overline h} \mathcal H_h\otimes \overline{\mathcal H_{\overline h}}\\[5pt]
&&\oplus\displaystyle \bigoplus_{n=1}^\infty 
\left[f_n \mathcal H_n\otimes \overline{\mathcal H_0}\oplus \overline{f_n} \mathcal H_0\otimes \overline{\mathcal H_n}\right]\\[5pt]
&\oplus&\displaystyle \bigoplus_{n=1}^\infty  
\left[g_n \mathcal H_n\otimes \overline{\mathcal H_{\rm mm}}\oplus \overline{g_n} \mathcal H_{\rm mm}\otimes \overline{\mathcal H_n}\right].
\end{array}
\end{equation}
The coefficients
$k_{h,\overline h},\, f_n,\,\overline f_n,\, g_n,\,\overline g_n$ are all non-negative integers,
and their precise values depend  on the specific K3 theory under inspection.
By  (\ref{ellgenform}), with (\ref{wittenindices}) and the
refined ansatz (\ref{refinedansatz}), we  obtain\footnote{Note that in \cite{eoty89}, the elliptic genus $\EEE^{\rm CFT}(\H;\tau,z)$
was already  decomposed 
in the spirit of (\ref{definitionofe}).}
\begin{equation}\label{definitionofe}
\left.
\begin{array}{rcl}
\displaystyle\EEE^{\rm CFT}(\H;\tau,z)
&=&\displaystyle
-2 \chi_0(\tau,z) +  20 \chi_{\rm mm}(\tau,z) + \sum_{n=1}^\infty \left[-2f_n +g_n\right]\chi_n(\tau,z) \\[12pt]
&=&\displaystyle
-2 \chi_0(\tau,z) +  20 \chi_{\rm mm}(\tau,z) + e(\tau)\,\widehat\chi(\tau,z)\\[5pt]
\displaystyle\mb{ with }\qquad\qquad
e(\tau)&=& \displaystyle\sum_{n=1}^\infty a_n q^n \;:= \;\sum_{n=1}^\infty \left[g_n-2f_n\right]q^n.
\end{array}
\right\}
\end{equation}
While the multiplicities $g_n$, $f_n$  vary 
within the moduli space of K3 theories, the coefficients 
$a_n:=g_n-2f_n$ 
of $e(\tau)$ are invariant. 
Since closed formulas for $\EEE^{\rm CFT}(\H;\tau,z)=\EEE(X;\tau,z)$, 
$\chi_0(\tau,z)$,  $\chi_{\rm mm}(\tau,z)$
and $\widehat\chi(\tau,z)$ are known
(see Def.~\ref{defk3cft} and Appendix \ref{N4characters}), one may solve the above equation 
(\ref{definitionofe}) for $e(\tau)$ if need be.

We will now make use of our Assumption \ref{K3chiralalgebra} for the 
first time, to argue that $a_n\geq0$ for all $n\in\N$. This follows, since the
spectral flow maps the irreducible representation $\mathcal H_0$ 
to the representation of the $N=4$ superconformal algebra whose ground state is
the vacuum. Hence the coefficients $f_n$ in (\ref{refinedansatz})
determine those contributions to
the space of states which are holomorphic but
do \textsl{not} belong to the vacuum representation under the  $N=4$
superconformal algebra. For any fixed value of $n\in\mathbb N$ with $n>0$, 
Assumption \ref{K3chiralalgebra} implies that
generically, no such additional contributions occur\footnote{This is 
independent of whether or not, within the moduli space,
one should expect that there is a dense subset of 
K3 theories  that possess some additional holomorphic states
beyond the vacuum representation of the $N=4$ superconformal
algebra.}.
In other words, our assumption implies that generically $f_n=0$ and thus that the 
$n^{\rm th}$ coefficient $a_n$ of $e(\tau)$ agrees with
$g_n\geq0$. Since  on the moduli space of K3 theories, 
these coefficients $a_n$
are invariant, $a_n\geq0$ follows. 

That the coefficients $a_n$ are non-negative was 
already conjectured in \cite{oo89} and independently in
\cite[Conj.7.2.2]{diss} and was later 
proved in \cite{eghi09,eot10}.
This gives evidence in favour of Assumption \ref{K3chiralalgebra}
to hold true. 
\medskip

In light of equation (\ref{genericsubspace}), we see that by definition,
and independently of Assumption \ref{K3chiralalgebra},
the decomposition (\ref{refinedansatz}) induces an isomorphism
of representations of the left-moving $N=4$ superconformal algebra,
extended by $\overline J_0$,
$$
\widetilde\H^R = {\mathcal H}_0\otimes \widetilde{\mathcal H}_0
\;\oplus\;
20 \mathcal H_{\rm mm}\otimes\widetilde{\mathcal H}_{\rm mm}
\;\oplus\;
\displaystyle \bigoplus_{n=1}^\infty 
f_n \mathcal H_n\otimes \widetilde {\mathcal H}_0
\;\oplus\; \bigoplus_{n=1}^\infty  
g_n \mathcal H_n\otimes\widetilde{\mathcal H}_{\rm mm},
$$
where $\widetilde{\mathcal H}_0,\,\widetilde{\mathcal H}_{\rm mm}$ denote the subspaces of Ramond
ground states in $\overline{\mathcal H_0},\,\qu{{\mathcal H}_{\rm mm}}$, respectively.
Therefore, and using $a_n\geq0$ for all $n\in\N$ with $n>0$ \cite{eghi09,eot10}, with
\begin{equation}\label{spaceproposal}
\widehat\H^R_{\rm min} :=
 \mathcal H_0\otimes \widetilde{\mathcal H}_0
\;\oplus\; 
20 \mathcal H_{\rm mm}\otimes\widetilde{\mathcal H}_{\rm mm} 
\;\oplus\; \bigoplus_{n=1}^\infty a_n  \mathcal H_n\otimes\widetilde{\mathcal H}_{\rm mm},
\end{equation}
we have
\begin{equation}\label{general}
\widetilde\H^R = \widehat\H^R_{\rm min} \;\oplus\; 
\bigoplus_{n=1}^\infty f_n 
\mathcal H_n\otimes\left(\widetilde{\mathcal H}_0 \;\oplus\; 2 \widetilde{\mathcal H}_{\rm mm}\right),
\end{equation}
with model dependent multiplicities $f_n\geq0$, such that independently
of Assumption \ref{K3chiralalgebra},
\begin{equation}\label{crucialcontained}
\widehat\H^R_{\rm min}\subset\widehat\H^R .
\end{equation}
Employing Assumption \ref{K3chiralalgebra}, the generic values of
$f_n,\, g_n$ across the entire moduli space of K3 theories
are $f_n=0$ and $g_n=a_n$, implying $\widehat\H^R =
\widehat\H^R_{\rm min}$. In fact,  we find
$$
\mb{Assumption \ref{K3chiralalgebra}}\quad\Longleftrightarrow\quad
\widehat\H^R =
\widehat\H^R_{\rm min}.
$$
Now recall that by construction,
\begin{eqnarray*}
\EEE^{\rm CFT}(\H;\tau,z) 
&\eet{(\ref{definitionofe}),\,(\ref{spaceproposal})}& \tr\nolimits_{\widehat\H^{R}_{\rm min}}\left( (-1)^{J_0-\overline J_0}y^{J_0} q^{L_0-{1\over4}}\right),\\
\EEE_{\rm Hodge}^{\rm CFT}(\H;\tau,z,\nu) 
&\eet{Def.~\ref{cftellgen}}& \tr\nolimits_{\widetilde\H^{R}}\left( (-1)^{J_0-\overline J_0}y^{J_0} u^{\overline J_0}q^{L_0-{1\over4}}\right).
\end{eqnarray*}
We introduce
$$
\EEE^{00}_{\rm Hodge}(X;\tau,z,\nu) 
:= \mbox{Tr}_{\widehat\H^{R}_{\rm min}}\left( (-1)^{J_0-\overline J_0}y^{J_0} u^{\overline J_0}q^{L_0-{1\over4}} \right),
$$
such that 
\begin{equation}\label{expectedHEG}
\mb{Assumption \ref{K3chiralalgebra}}\quad\Longleftrightarrow\quad
\EEE^0_{\rm Hodge}(X;\tau,z,\nu) = \EEE^{00}_{\rm Hodge}(X;\tau,z,\nu) .
\end{equation}
Here, $\widetilde\H^R$, $\widehat\H^R$ and $\widehat\H^R_{\rm min}$ are
solely viewed as representations of
the left-moving $N=4$ superconformal algebra, extended by $\qu J_0$.
A failure of  (\ref{expectedHEG})
would  imply that K3 theories generically possess a  chiral algebra which is extended
beyond the $N=4$ superconformal algebra at central charge $c=6$. In fact, by the discussion
of Sect.~\ref{genericity}, all common eigenspaces of $J_0,\, \qu J_0$ and $L_0$ in
$\wt\H^R$ attain the minimal possible dimensions across the moduli space.
A failure of  (\ref{expectedHEG}) would thus mean that \textsl{all} K3 theories possess
a chiral algebra which is a proper extension of the $N=4$ superconformal algebra at 
central charge $c=6$, which would in fact be quite
an exciting result. 
We emphasize once again
that we do {\em not} assume that there might not be a dense subset of 
K3 theories within the moduli space that do possess holomorphic
states beyond the vacuum representation of the $N=(4,4)$ superconformal 
algebra. 
Indeed, we see no reason to make such a strong assumption, the 
analogue of which fails for toroidal superconformal field theories, as we have
seen in Sect.~\ref{warm}.

If one wishes to make a connection to Mathieu Moonshine, one may replace the
multiplicities $a_n$ by representations of $M_{24}$ according to \cite{ga12},
see also \cite[\S4]{we14}. We will come back to this comment in Sect.~\ref{Mathieu}.
\medskip

It is now straightforward to calculate a 
closed formula for $\EEE^{00}_{\rm Hodge}(X;\tau,z,\nu)$:
by (\ref{spaceproposal}) and using (\ref{leadingorder}), 

$$
\EEE^{00}_{\rm Hodge}(X;\tau,z,\nu) 
=\chi_0(\tau,z) \cdot (-u-u^{-1})
+  20 \chi_{\rm mm}(\tau,z) + \sum_{n=1}^\infty a_n \chi_n(\tau,z).
$$
Inserting (\ref{definitionofe}), 
we arrive at 
\begin{equation}\label{minimalgenera}
\EEE_{\rm Hodge}^{00}(X;\tau,z,\nu) 
=
(2-u-u^{-1})\cdot \chi_0(\tau,z) + \EEE(X;\tau,z),
\end{equation}
where closed formulas for $\chi_0(\tau,z)$ and $\EEE(X;\tau,z)$ are given
in Appendix \ref{N4characters} and \mb{(\ref{K3ellgen})}, respectively.
Note that $\chi_0(\tau,z)$ has Mock modular transformation properties, as
mentioned above.
For an arbitrary K3 theory,  (\ref{general}) yields
$$
\EEE_{\rm Hodge}^{\rm CFT}(\H;\tau,z,\nu) 
= \EEE^{00}_{\rm Hodge}(X;\tau,z,\nu) + (2-u-u^{-1})\cdot \sum_{n=1}^\infty f_n \chi_n(\tau,z),
$$
with model dependent multiplicities $f_n\geq0$. In fact,
the above derivation proves
\begin{proposition}\label{HEGforK3}
Let $X$ denote a K3 surface, and consider a K3 theory 
according to Def.~\mb{\ref{defk3cft}} with space of 
states $\H$, whose Ramond sector $\H^R$ decomposes according to \mb{(\ref{refinedansatz})}. 
Then the conformal field theoretic Hodge-elliptic genus of this theory is
given by
$$
\EEE_{\rm Hodge}^{\rm CFT}(\H;\tau,z,\nu) 
=  (2-u-u^{-1})\cdot \left(\chi_0(\tau,z) + \sum_{n=1}^\infty f_n \chi_n(\tau,z)\right)
+ \EEE(X;\tau,z) .
$$
Furthermore,  precisely one of the following holds:
either, the generic conformal field theoretic 
Hodge-elliptic genus of all K3 theories obeys 
$$
\EEE^0_{\rm Hodge}(X;\tau,z,\nu) 
=
(2-u-u^{-1})\cdot \chi_0(\tau,z) + \EEE(X;\tau,z),
$$
or all K3 theories possess a   chiral algebra which is
a proper extension of the $N=4$ superconformal algebra at central
charge $c=6$.
\end{proposition}
Recall that  all standard $\Z_2$-orbifold conformal field theories obtained
from non-linear sigma models with target a complex two-torus 
possess a chiral algebra which is a strict enhancement
of the $N=4$ superconformal algebra at central charge $c=6$.
This is immediately reflected in the fact that the formula for the
generic conformal field theoretic Hodge-elliptic genus in 
Prop.~\ref{HEGforK3} differs from the formula stated in Prop.~\ref{HEGforZ2orbs}.
On first sight, this observation may seem surprising, because the 
latter formula is obtained from the Hodge-elliptic
genus of complex two-tori (\ref{toruskatr}) by standard
orbifold techniques. Hence the discrepancy between the formula in Prop.~\ref{HEGforK3} and
the orbifold one implies that orbifolding techniques do not
apply to the calculation of the generic conformal field theoretic 
Hodge-elliptic genus. Indeed, even though 
the partition function as well as
the elliptic genus for K3 theories may be obtained
by this procedure from the respective quantities for
two-tori, this idea cannot work for the Hodge-elliptic
genus. The explanation lies in the action of the modular group: 
while orbifolding techniques both for the partition function
and the elliptic genus heavily use the fact that both
of them exhibit modular transformation properties, 
there is no reason to expect such modular behaviour
for the generic conformal field theoretic Hodge-elliptic genus.
For K3 theories, Prop.~\ref{proposedHEG} unveils a Mock
modular behaviour, instead, albeit elliptic in $z$ with 
respect to $\Lambda_\tau=\Z\tau+\Z$.
\subsection{The complex Hodge-elliptic genus of K3 surfaces}\label{HodgeK3}
How the complex Hodge-elliptic genus of \cite{katr16} can be of help
in the investigation of K3 theories or sigma models is not clear to us. But
it certainly yields a new, highly non-trivial and  interesting  
invariant for K3 surfaces, which is investigated more closely
in this section. Though we will leave it to future work to achieve
this goal, one important aim is to state a closed formula
for the complex Hodge-elliptic genus of K3 surfaces.
\smallskip

In the following, let $X$ denote a K3 surface. 
Recall that the virtual bundle $\E_{q,-y}$ on $X$, which crucially enters
the Definition \ref{geoell} of the complex Hodge-elliptic genus, according
to \cite[\S4, Conjecture 1]{we14} allows a decomposition which induces
the decomposition (\ref{definitionofe})
of the complex elliptic genus into characters of irreducible
representations of the $N=4$ superconformal algebra: with notations as in
Def.~\ref{geoell},
\begin{equation}\label{Edeco}
\mathbb E_{q,-y} 
= -\mathcal O_X\cdot \chi_0(\tau,z) -  T\cdot \chi_{\rm mm}(\tau,z) + \sum_{n=1}^\infty p_n(T)\cdot \chi_n(\tau,z),
\end{equation}
where $p_n(T)$ is a virtual bundle of the form 
\begin{equation}\label{anasEuler}
p_n(T)={\sum\limits_{k=0}^{N_n}} \alpha_k T^{\otimes k}, \quad \alpha_k\in\Z, 
\quad \mb{ with }\quad a_n = \chi\left( p_n(T) \right),
\; \mb{ where }\; e(\tau)=\sum\limits_{n=1}^\infty a_n q^n
\end{equation}
as before, in 
(\ref{definitionofe}). This can be viewed 
as a generalization of a {\em local index theorem} 
\cite{pa71,gi73,abp73,ge83}. With Thomas Creutzig \cite{crwe15}, 
we have  proved the claims (\ref{Edeco}), (\ref{anasEuler}), using ideas that
have been developed in \cite{crho13}. In fact, our proof reveals a refinement
of the above conjecture: up to a global sign, 
all the virtual bundles $p_n(T)$ turn out
to be direct sums of symmetric tensor powers of the holomorphic tangent
bundle $T$, thus yielding each $-p_n(T)$ as an honest holomorphic vector bundle
rather than a {\em virtual} bundle. 
For the complex Hodge-elliptic genus of K3 we thus obtain 
\begin{align}\label{HodgeK3preformula}
 \EEE_{\rm Hodge}&(X;\tau,z,\nu) \nonumber\\
& \eet{Def.~\ref{geoell}} u^{-1} \chi^u( \mathbb E_{q,-y} )\nonumber\\
& \ee{Edeco}
-u^{-1} \chi^u(\mathcal O_X)\cdot \chi_0(\tau,z) 
-  u^{-1} \chi^u(T) \cdot \chi_{\rm mm}(\tau,z) 
+ \sum_{n=1}^\infty u^{-1} \chi^u(p_n(T))\cdot \chi_n(\tau,z)\nonumber\\
&= - (u^{-1}+u) \cdot \chi_0(\tau,z) +20 \chi_{\rm mm}(\tau,z) 
+ \sum_{n=1}^\infty u^{-1} \chi^u(p_n(T))\cdot \chi_n(\tau,z).
\end{align}
Equipped with this information, we are now ready to prove
\begin{proposition}\label{chiralisnotHodge}
Let $X$ denote a K3 surface, and suppose that Assumption 
\mb{\ref{K3chiralalgebra}} holds. 
Then the complex Hodge-elliptic genus of $X$ differs from the generic
conformal field theoretic Hodge-elliptic genus of K3 theories.
\end{proposition}
\begin{proof}
We prove the claim by contradiction. 

\noindent
So let us assume that 
$\EEE_{\rm Hodge}^0(X;\tau,z,\nu)=\EEE_{\rm Hodge}(X;\tau,z,\nu)$.
Since Assumption \ref{K3chiralalgebra} is supposed to hold,
Prop.~\ref{HEGforK3} yields a closed formula for 
$\EEE_{\rm Hodge}^0(X;\tau,z,\nu)$. 
Comparison with the formula (\ref{HodgeK3preformula}) 
for the Hodge-elliptic genus of $X$ implies that for all $n\in\N$ with $n>0$, we have
$\chi\left( p_n(T) \right) = a_n =u^{-1} \chi^u(p_n(T))$,
in other words, that the holomorphic vector bundle 
$E_n:=-p_n(T)$ obeys
$H^j(X,E_n) = \{0\}$ if $j\neq1$. However, as one checks
by a direct calculation, $E_1=S^2(T)$, $E_2=2S^3(T)\oplus{\mathcal O}_X$,
and hence $H^0(X,E_2) \neq \{0\}$.
\end{proof}
In principle, (\ref{HodgeK3preformula}) yields a formula for the complex Hodge
elliptic genus of K3 surfaces. However, so far, no closed formula for the vector
bundles $E_n$, $n\in\N$, is known. To arrive at a less implicit 
presentation, note  that the derivation of (\ref{HodgeK3preformula})
solely uses the observation that the virtual bundle $\E_{q,-y}$ possesses global
holomorphic sections which on each fiber of this bundle yield the
structure of an $N=4$ superconformal algebra at central charge $c=6$.
But the crucial step in the above proof of Prop.~\ref{chiralisnotHodge} rests
on the fact that the virtual bundle $\E_{q,-y}$  possesses
additional global holomorphic sections. In fact,  
\cite[Thm.~3.1 and Remark 3.6]{crho13} imply that 
the global holomorphic sections of this bundle induce the structure of a  
module on each fiber of $\E_{q,-y}$ for a super vertex operator algebra
$V^{SU(2)}$ which extends the $N=4$ superconformal algebra by $8$ fields of 
conformal dimensions $2,\,{5\over2},\,{5\over2},\,{5\over2},\,{5\over2},\, 3,\, 3,\, 3$,
respectively. The explicit form of these fields is also stated in
\cite[Remark 3.6]{crho13}. This additional structure of the virtual bundle
$\E_{q,-y}$ may be used to arrive at a more promising formula
for the complex Hodge-elliptic genus. So far, it remains implicit, as we shall explain
next.
\begin{remark}\label{HEGformulaimplicit}
According to \cite[Thm.~3.12]{crho13}, similarly to the situation
for the $N=4$ superconformal algebra at
central charge $c=6$, the super vertex operator algebra $V^{SU(2)}$
possesses three types of irreducible representations with 
characters $\mb{ch}_m, \, m\in\N$, where $m\in\{0,\,1\}$ gives 
two {\em massless} characters and $m>1$ gives  
an infinite family of {\em massive} ones. These characters are 
related to those of the $N=4$ superconformal algebra via
\begin{eqnarray*}
\mb{ch}_0(\tau,z) \;=\; \chi_0(\tau,z) + H_0(\tau)\cdot\widehat\chi(\tau,z),\qquad\qquad
\mb{ch}_1(\tau,z) &=& \chi_{\rm mm}(\tau,z) + H_1(\tau)\cdot\widehat\chi(\tau,z),\\[3pt]
\forall m>1\colon\qquad
\mb{ch}_m(\tau,z) &=& H_m(\tau)\cdot\widehat\chi(\tau,z),
\end{eqnarray*}
where
\begin{eqnarray*}
\forall m\in\N\colon\;\;
H_m(\tau) &=& h_m(\tau)-2h_{m+1}(\tau)+2h_{m+3}(\tau)-h_{m+4}(\tau)+\delta_{m,0},\\
h_{m}(\tau) &=& {\textstyle{1\over\eta(\tau)^3}} 
\sum_{\stackrel{k,r,s\in\Z+{1\over2}}{ r,s>0}} (-1)^{r+s+1} 
q^{r|k|+s|m-1-k|+{1\over2}(\sgn(k)r+\sgn(k-m+1)s)^2-{m-1\over2}},
\end{eqnarray*}
again by \cite[Thm.~3.12]{crho13}. In particular, the Witten genera of
these characters agree with those of the respective
characters of the $N=4$
superconformal algebra.  

Applying the same arguments that prove the decomposition (\ref{HodgeK3preformula}),
one obtains a refined decomposition of the bundle $\E_{q,-y}$ of the form
$$
\mathbb E_{q,-y} 
= -\mathcal O_X\cdot \mb{ch}_0(\tau,z) -  T\cdot \mb{ch}_1(\tau,z) 
- \sum_{m=2}^\infty P_m(T)\cdot \mb{ch}_m(\tau,z),
$$
where $P_m(T)$ is a virtual bundle of the form 
$P_m(T)={\sum\limits_{k=0}^{N_m}} \beta_k T^{\otimes k}$, $\beta_k\in\Z$.
Moreover, \cite[Prop.~3.10]{crho13} implies that $P_m(T)$ has no
 non-trivial global holomorphic sections. We therefore find the following formula
for the Hodge-elliptic genus of any K3 surface $X$:
$$
\EEE_{\rm Hodge}(X;\tau,z,\nu)
= -(u+u^{-1}) \mb{ch}_0(\tau,z) +20 \mb{ch}_1(\tau,z) 
+ \sum_{m=2}^\infty b_m \mb{ch}_m(\tau,z), 
$$
where the coefficients $b_m=\dim H^1(X,P_m(T))$ are obtained by comparison 
to (\ref{definitionofe}), which implies
$$
\sum_{m=2}^\infty b_m H_m(\tau) = e(\tau)+2H_0(\tau)-20H_1(\tau).
$$
Note that our formula for the complex Hodge-elliptic genus of K3 surfaces
is independent of the choice of complex
structure, as it should. It would be interesting to know its modular
properties.
\end{remark}
\subsection{The chiral Hodge-elliptic genus of K3}\label{comments}
In this section, we prove an explicit formula
for the chiral Hodge-elliptic genus
of K3 surfaces. 
\smallskip

First note that Prop.~\ref{HEGforK3} together with
Conjecture \ref{genericischiral}  immediately
implies a conjectural formula for the chiral Hodge-elliptic genus,
stated in Prop.~\ref{proposedHEG}, below. 
In fact, the approach given here already proves this
formula if Assumption \ref{K3chiralalgebra} 
and Conjecture \ref{genericischiral} hold true. That this 
latter conjecture is satisfied follows from the identification
of the infinite volume limit of the topologically  half-twisted sigma model
on a K3 surface $X$ with the sheaf cohomology of the chiral de Rham complex
$\Omega_X^{\rm ch}$ by Kapustin \cite{ka05}. Since Assumption \ref{K3chiralalgebra}
is expected to hold true, as well, the above already
gives a derivation of this formula from string theory ingredients.  
In addition, we give a direct mathematical proof of this proposition, below.
\begin{proposition}\label{proposedHEG}
The chiral Hodge-elliptic genus of a K3 surface X \mb{(}Def.~\mb{\ref{chiralHEG}}\mb{)} is given
by
$$
\EEE_{\rm Hodge}^{\rm ch}(X;\tau,z,\nu) 
=
(2-u-u^{-1})\cdot \chi_0(\tau,z) + \EEE(X;\tau,z),
$$
where closed formulas for $\chi_0(\tau,z)$ and $\EEE(X;\tau,z)$ can
be found
in Appendix \mb{\ref{N4characters}} and \mb{(\ref{K3ellgen})}, 
respectively. In particular, if Assumption \mb{\ref{K3chiralalgebra}}
holds, then Conjecture \mb{\ref{genericischiral}} is true, i.e.\
for K3 theories we have
$$
\EEE_{\rm Hodge}^{\rm ch}(X;\tau,z,\nu) 
= \EEE_{\rm Hodge}^0 (\H;\tau,z,\nu) .
$$
\end{proposition}
\begin{proof}
The second statement follows from the first one by Prop.~\ref{HEGforK3},
while the first statement is almost immediate from Bailin Song's result
\cite[Thm.~1.2]{so16}
that the  chiral de Rham complex
for a K3 surface does not have global holomorphic sections 
other than those furnishing the $N=4$ superconformal 
vertex operator algebra
at central charge $c=6$, which were established in \cite{bhs08,he09}.
See also \cite[\S4]{so17}, where 
all relevant steps of the following proof can be found, as well. 

Indeed, from \cite[Thm.~1.2]{so16} and using the notations introduced in
Sect.~\ref{genericK3CFT}, above, 
it follows that as a representation of the $N=4$ superconformal algebra,
$H^0(X,\Omega^{ch}_X)\cong\HHH_0$. 
Moreover, Poincar\'e duality holds for the chiral de Rham complex
\cite{masc99}, hence $H^2(X,\Omega^{ch}_X)\cong\HHH_0$, as well.
Extending the $N=4$ superconformal algebra by $\qu J_0$, more
precisely we have 
$H^0(X,\Omega^{ch}_X)\oplus H^2(X,\Omega^{ch}_X)\cong\HHH_0\otimes\widetilde\HHH_0$.
Thus by Def.~\ref{chiralHEG}, 
$$
\EEE^{\rm ch}_{\rm Hodge}(X;\tau,z,\nu)
= -(u+u^{-1}) \chi_0(\tau,z)
- y^{-{1}} \tr\nolimits_{H^1(X,\Omega^{\rm ch}_X)} 
\left( (-1)^{J_0} y^{J_0}  q^{L_0^{\rm top}} \right).
$$
Insertion of $\nu=0$ according to (\ref{chHEG0EG}) yields 
the complex elliptic genus $\EEE(X;\tau,z)$, and since the 
only unknown contributions to the right hand side of the
above formula are independent of $u$, one can solve
for those unknown contributions at $\nu=0$, confirming the claim.
\end{proof}
Reversing the above arguments, one may view the string theory
derivation of Prop.~\ref{proposedHEG} given previously as an alternative, 
not entirely mathematical derivation of the beautiful result 
\cite[Thm.~1.2]{so16} that the  global holomorphic sections of the chiral de Rham complex
for K3 surfaces yield precisely the $N=4$ superconformal 
vertex operator algebra
at central charge $c=6$. 
Moreover, we have
\begin{proposition}\label{assumptionproved}
If the infinite volume limit of a topologically half-twisted sigma model
on a K3 surface $X$ yields the sheaf cohomology of the chiral de Rham complex 
of $X$, as is argued in \mb{\cite{ka05}}, then the generic
chiral algebra of K3 theories is the $N=4$ 
superconformal vertex operator algebra at central charge
$c=6$.
\end{proposition}
\begin{proof}
Our assumption on the infinite volume limit of a topologically half-twisted
sigma model on a K3 surface $X$ implies that $H^\ast(X,\Omega_X^{ch})$
is the generic space of states of K3 theories at generic volume and some
fixed choice of a hyperk\"ahler structure and B-field on $X$. On the other
hand, $\widehat\H^R$ is the generic space of states at generic values of
all moduli of K3 theories. Hence we have
$$
\widehat\H^R_{\rm min}
\quad\stackrel{\mbox{\scriptsize\rm\eqref{crucialcontained}}}{\subset}
\quad\widehat\H^R \quad
\subset \quad H^\ast(X,\Omega_X^{ch}).
$$
Note that this even holds if the moduli space of K3 theories 
possesses more than one component, since by construction,
the number $N_{h,Q,\qu Q}$ in Def.~\ref{genericCFTHEG} is obtained
as infimum over the entire moduli space, and $\widehat\H^R_{\rm min}$ is constructed
as to obey \eqref{crucialcontained}. 
By \eqref{minimalgenera} and Prop.~\ref{proposedHEG}, 
the full characters of 
$\widehat\H^R_{\rm min}$ and $H^\ast(X,\Omega_X^{ch})$ agree, hence
$\widehat\H^R_{\rm min}=H^\ast(X,\Omega_X^{ch})$ and thus 
$\widehat\H^R_{\rm min}=\widehat\H^R$. We conclude that 
$\EEE^0_{\rm Hodge}(X;\tau,z,\nu) = \EEE^{00}_{\rm Hodge}(X;\tau,z,\nu)$
in \eqref{minimalgenera} and thus, by Prop.~\ref{HEGforK3},
that Assumption \ref{K3chiralalgebra} holds, as claimed.
\end{proof}
The above results explain why, in contrast to the chiral de Rham complex, 
the virtual bundle 
$\E_{q,-y}$ can only be of limited use for the investigation of the fine
structure of K3 theories. Indeed, by what was said in Sect.~\ref{HodgeK3},
$\E_{q,-y}$ has global holomorphic sections that do not belong to the $N=4$
superconformal vertex operator algebra at central charge $c=6$.
This is a profound difference to the sheaf cohomology of the
chiral de Rham complex,  which in turn seems to beautifully model
a generic field content of all K3 theories.

Note furthermore that by Prop.~\ref{proposedHEG},
the chiral Hodge-elliptic genus of K3 surfaces is independent of the complex
structure, i.e.\ it yields a new topological invariant in this case.
Kapustin's work only assumes the transition to an infinite volume limit. We
have thus shown the suprising fact that in the case of K3 theories, such an
infinite volume limit, viewed solely as a representation of the $N=4$ superconformal
algebra extended by $\qu J_0$, yields a generic space of states for all
K3 theories. Moreover,
the result of Prop.~\ref{proposedHEG} is compatible
with Conjecture \ref{mirror}, as one immediately checks, since the mirror of 
a K3 surface is a K3 surface.
\subsection{A geometric  Mathieu Moonshine Module}\label{Mathieu}
We close this note by commenting on the consequences 
of our findings for Mathieu Moonshine\footnote{Some aspects 
of this discussion can be found 
analogously in the preprint \cite{so17} by Bailin Song,
which reached me during the final stages of writing this note,
as mentioned already in the Introduction.}.
To this end, consider a K3 surface $X$ with 
fixed  complex structure.
Then any finite symplectic automorphism group $G$ of 
$X$ has a natural induced action on the cohomology
$H^\ast( X, \E_{q,-y})$
of the virtual bundle $\E_{q,-y}$, but also on the
cohomology of the chiral de Rham complex, according to
\cite[(2.1.3)]{goma03}. By \cite[Thm.~0.3]{mu88}, 
$G$ is a subgroup of the one-point stabilizer
$M_{23}$ of the Mathieu group $M_{24}$.
{\em Mathieu Moonshine}, on the other hand, 
which was discovered by Eguchi, Ooguri and Ta\-chi\-ka\-wa in
\cite{eot10} and proved by Gannon in \cite{ga12}, 
predicts that the decomposition (\ref{definitionofe})
into characters of irreducible representations of the
$N=4$ superconformal algebra is governed by an
action of $M_{24}$. Yet what $M_{24}$ should act on, has
been a mystery, so far.

Up to now, $H^\ast( X, \E_{q,-y})$ and 
$H^\ast(X,\Omega_X^{\rm ch})$ seemed equally
auspicious  to this effect, since both carry natural
actions of an $N=4$ superconformal algebra
at central charge $c=6$, and of the finite symplectic
automorphism group $G$. 
By \cite[Thm.~4.3]{crho13} and 
\cite[Thm.~3.3]{so17}, the resulting twisted elliptic genera
agree with those of Mathieu Moonshine. 
However,
Prop.~\ref{proposedHEG} now shows that 
$H^\ast(X,\Omega_X^{\rm ch})$ is by far
more promising. 
Indeed, it implies that as a module of
the $N=4$ superconformal algebra, 
none of the contributions coming from
massive representations contains a {\em virtual}  representation,
in agreement with the results of \cite{ga12}. As was explained
in Sect.~\ref{HodgeK3}, this is different for 
$H^\ast( X, \E_{q,-y})$. Since the focus
of \cite{crho13} is entirely on the fine structure
of $\E_{q,-y}$, this may clarify why their attempts
of explaining Mathieu Moonshine failed.
We also conclude that our 
conjecture \cite[Conjecture 1]{we14} is 
a red herring when it comes to  a geometric realization
of the representation of $M_{24}$ that is relevant for 
Mathieu Moonshine. However, since this conjecture is correct
\cite{crwe15}, we hope that it may prove useful in the
study of the {\em complex Hodge-elliptic genus} of \cite{katr16}, instead. 
On the other hand, the above findings support the expectations that
we stated in \cite[\S4.2]{we14}, namely that the (holomorphic) chiral
de Rham complex might bear the key to understanding
Mathieu Moonshine. 
\smallskip

To arrive at a satisfactory explanation for Mathieu Moonshine,
one must now find a natural way to equip 
$H^\ast( X, \Omega_X^{\rm ch})$ with all the 
structures predicted by Mathieu Moonshine. 

As was pointed out above, 
for any choice of complex structure
on our K3 surface,  the corresponding finite symplectic 
automorphism groups act naturally on 
$H^\ast( X, \Omega_X^{\rm ch})$ in a fashion that is
compatible with Mathieu Moonshine.
One must now find a way to combine the actions of 
all such finite symplectic
automorphism groups to the action of $M_{24}$. 
To do so,  in \cite{tawe11,tawe12,tawe13}
Taormina and the author
have proposed a technique
called \textsc{symmetry surfing}.
Focusing on standard $\Z_2$-orbifold conformal field theories obtained
from non-linear sigma models with target a complex two-torus, 
in these works we have been able to show that symmetry
surfing allows to combine all finite symplectic automorphism groups
of Kummer surfaces to the maximal subgroup
$\mathbb Z_2^4\rtimes A_8$ of $M_{24}$. 

Furthermore, in \cite{tawe13} Taormina and the author show that
the action of this group on the leading order massive representation may
be realized on a subspace of the space of states that is common to all 
standard $\Z_2$-orbifold conformal field theories obtained
from non-linear sigma models with target a complex two-torus.
The resulting representation
is equivalent to the restriction of the corresponding Mathieu 
Moonshine action of $M_{24}$ to this subgroup. 
The symmetry groups from distinct points in
moduli space must be combined with a {\em twist}. 
Further evidence in favour of symmetry surfing, including the twist,
is provided in \cite{gakepa16}. 
The focus on a subspace of the space of states which is
common to all K3 theories that have been accessible to 
these methods, so far, is in full accord with the 
expectation that the 
cohomology of the chiral de Rham complex might play
a key role in the explanation of Mathieu Moonshine. 
The behaviour of the chiral Hodge-elliptic genus
found in this note further  supports this idea.  
Indeed, we view the chiral Hodge-elliptic genus as a refinement of
the traditional complex elliptic genus. Its
agreement with the generic conformal field theoretic
Hodge-elliptic genus, addressed in
Props.~\ref{proposedHEG}
and \ref{assumptionproved}, supports the idea that the sheaf cohomology
$H^\ast(X,\Omega_X^{\rm ch})$ of the chiral de Rham complex should be viewed as 
a model for a subspace of the space of states that
is {\em generically} present in K3 theories, along the lines
presented in Sect.~\ref{genericK3CFT}.
Whether or not this subspace varies smoothly with respect
to the moduli remains unanswered, for the time being.
As is emphasized in \cite{tawe11,tawe12,tawe13}, 
symmetry surfing crucially requires to restrict
attention to the {\em geometric} symmetry groups, i.e.\
to groups $G$ that arise as finite symplectic automorphism
groups of K3 surfaces\footnote{To lift the action of $G$ to a K3 theory
with geometric interpretation on the respective K3 surface, one may
allow a non-trivial B-field iff the latter can be represented by some 
$G$-invariant $B\in H^2(K3,\R)$.}. We have already pointed
out in \cite[\S4.2]{we14} that this might have its explanation
in a required compatibility with an infinite
 volume limit of topologically half-twisted K3 theories, where we expect
to find the sheaf cohomology of the chiral de Rham complex
by Kapustin's claims \cite{ka05}.

It remains an open problem,
however, to extend the action of the maximal subgroup 
$\mathbb Z_2^4\rtimes A_8$ to an 
action of the entire group $M_{24}$. Then, an interpretation 
must be found for 
the action of those elements of $M_{24}$ which cannot act
as finite symplectic automorphisms on any K3 surface.
Furthermore, the behaviour
of the multiplicity spaces of the {\em massless} representations of 
the $N=4$ superconformal algebra remains obscure. 
Not least, why of all groups the Mathieu group $M_{24}$ plays such
a prominent role for K3, remains unknown.
\smallskip

Finally, Mathieu Moonshine predicts the structure of a vertex operator
algebra on the representation space that underlies 
$H^\ast( X, \Omega_X^{\rm ch})$. Indeed, according to
\cite[Prop.~3.7 and Def.~4.1]{bo01}, the cohomology of the chiral
de Rham complex bears the structure of a super vertex operator
algebra. For the $\Z_2$-orbifold conformal field theories 
obtained from non-linear sigma models with target a complex two-torus, 
the results of \cite{boli00a,frsz07,tawe12,gakepa16}  give strong
evidence in favour of compatibility with
the combined symmetry group $\mathbb Z_2^4\rtimes A_8$ of $M_{24}$,
if one respects the twist. For any K3 surface $X$, 
by introducing a novel filtration on $H^\ast( X, \Omega_X^{\rm ch})$, Song proves 
in \cite[Thm.~3.2]{so17} that the associated graded object
 is a unitary representation of
the $N=4$ superconformal vertex operator algebra at central
charge $c=6$.
\smallskip

Altogether,    Mathieu Moonshine seems to gradually
unveil its mysteries.
\bigskip

{\textbf{Acknowledgements}}

It is my great pleasure to thank Thomas Creutzig,
Shamit Kachru, Stefan Kebekus,
Anatoly Libgober, Emanuel Scheidegger and Anne Taormina  
for very helpful communications and discussions. 
I particularly thank Bailin Song for his comments on an
earlier version of this note, which
lead to the formulation of Prop.~\ref{chiralisnotHodge}
and to a considerable extension of the interpretation of
my calculations. I am grateful to an anonymous referee
for their diligent reading of the manuscript and constructive
criticism.

I am also grateful to the organisers of the program on {\em Automorphic forms, mock modular forms 
and string theory} in September 2016 and to the Simons Center for Geometry and Physics at Stony Brook
for the support and hospitality, and for the inspiring environment during 
the initial steps of this work.
\appendix
\section{Some $N=4$ characters}\label{N4characters}
The characters of the irreducible representations of the $N=4$ superconformal
algebra have been determined explicitly in \cite{egta88a}. Here, we restrict
ourselves to stating the characters in the twisted Ramond sector $\widetilde{\rm R}$
at central charge $c=6$. 

For our purposes, the most convenient formulas 
use the standard Jacobi theta functions, where we take the following 
normalizations:
\begin{eqnarray*}
\vartheta_1(\tau,z) & := &
i\sum_{n=-\infty}^\infty (-1)^n q^{{1\over2}(n-{1\over2})^2} y^{n-{1\over2}}\\
&=& iq^{{1\over8}} y^{-{1\over2}} \prod_{n=1}^\infty (1-q^n)(1-q^{n-1}y)(1-q^{n}y^{-1}),\nonumber\\
\vartheta_2(\tau,z) & := &
\sum_{n=-\infty}^\infty q^{{1\over2}(n-{1\over2})^2} y^{n-{1\over2}}\\
&=&q^{{1\over8}} y^{-{1\over2}} \prod_{n=1}^\infty (1-q^n)(1+q^{n-1}y)(1+q^{n}y^{-1}),\nonumber\\
\vartheta_3(\tau,z) & := &
\sum_{n=-\infty}^\infty q^{{n^2\over2}} y^{n}\\
&=&\prod_{n=1}^\infty (1-q^n)(1+q^{n-{1\over2}}y)(1+q^{n-{1\over2}}y^{-1}),
\\
\vartheta_4(\tau,z) & := &
\sum_{n=-\infty}^\infty (-1)^n q^{{n^2\over2}} y^{n}\\
&=&\prod_{n=1}^\infty (1-q^n)(1-q^{n-{1\over2}}y)(1-q^{n-{1\over2}}y^{-1}).
\end{eqnarray*}
Moreover, we need  the Mordell function 
$h_3(\tau)$ \cite[p.~347]{mo33},
$$
h_3(\tau) :=
{1\over\eta(\tau)\vartheta_3(\tau,0)} \sum_{m\in\mathbb{Z}} {q^{{m^2\over2}-{1\over8}}\over 1+q^{m-{1\over2}}}
=
{2\over\eta(\tau)\vartheta_3(\tau,0)} \sum_{m\in\mathbb{N}\setminus\{0\}} {q^{{m^2\over2}-{1\over8}}\over 1+q^{m-{1\over2}}}.
$$
Then,  including the more commonly used notations for 
those characters,
\begin{eqnarray*}
\chi_0(\tau,z) &=& \mbox{ch}^{\widetilde R}_0(\ell=0;\tau,z)\\
&=& -2\left( {\vartheta_3(\tau,z)\over\vartheta_3(\tau,0)} \right)^2
+ \left( {q^{-{1\over8}}\over\eta(\tau)} - 2h_3(\tau)\right) \cdot
\left( {\vartheta_1(\tau,z)\over\eta(\tau)} \right)^2\\
\chi_{\rm mm}(\tau,z) &=& \mbox{ch}^{\widetilde R}_0(\ell={\textstyle{1\over2}};\tau,z)\\
&=& \left( {\vartheta_3(\tau,z)\over\vartheta_3(\tau,0)} \right)^2
+ h_3(\tau) 
\left( {\vartheta_1(\tau,z)\over\eta(\tau)} \right)^2,\\
\widehat\chi(\tau,z) &=& \mbox{ch}^{\widetilde R}(\ell={\textstyle{1\over2}};\tau,z)\\
&=& {q^{-{1\over8}}\over\eta(\tau)} \left( {\vartheta_1(\tau,z)\over\eta(\tau)} \right)^2.
\end{eqnarray*}
From these formulas, one reads the leading order contributions 
in $q$ of each character:
\begin{equation}\label{leadingorder}
\begin{array}{rcl}
\chi_0(\tau,z) = -y-y^{-1} + {\mathcal O}(q), \quad
\chi_{\rm mm}(\tau,z) &=& 1 + {\mathcal O}(q), \;\\[5pt]
\widehat{\chi}(\tau,z) &=& 2-y-y^{-1} + {\mathcal O}(q).
\end{array}
\end{equation}
%

\begin{thebibliography}{AKMW87}

\bibitem[ABD{\etalchar{+}}76]{aetal76}
{\sc M.~Ademollo, L.~Brink, A.~D'Adda, R.~D'Auria, E.~Napolitano, S.~Sciuto,
  E.~{Del Giudice}, P.~{Di Vecchia}, S.~Ferrara, F.~Gliozzi, R.~Musto, and
  R.~Pettorino}, \emph{Supersymmetric strings and color confinement}, Phys.
  Lett. \textbf{B62} (1976), 105--110.

\bibitem[ABP73]{abp73}
{\sc M.~Atiyah, R.~Bott, and V.K. Patodi}, \emph{On the heat equation and the
  index theorem}, Invent. Math. \textbf{19} (1973), 279--330, Errata: Invent.
  Math. {\bf 28}, 277--280 (1975).

\bibitem[AKMW87]{akmw87}
{\sc O.~Alvarez, T.P. Killingback, M.~Mangano, and P.~Windey}, \emph{String
  theory and loop space index theorems}, Commun. Math. Phys. \textbf{111}
  (1987), 1--12.

\bibitem[AM94]{asmo94}
{\sc P.S. Aspinwall and D.R. Morrison}, \emph{String theory on {K}3 surfaces},
  in: Mirror symmetry II, B.~Greene and S.T. Yau, eds., AMS, 1994,
  pp.~703--716; {\tt arXiv: hep-th/9404151}.

\bibitem[AMP12]{amp12}
{\sc P.S. Aspinwall, I.V. Melnikov, and M.R. Plesser}, \emph{{{\rm(0,2)}
  Elephants}}, JHEP \textbf{01} (2012), 060; {\tt arXiv: 1008.2156 [hep-th]}.

\bibitem[BHS08]{bhs08}
{\sc D.~{Ben-Zvi}, R.~{Heluani}, and M.~{Szczesny}}, \emph{{Supersymmetry of
  the chiral de Rham complex}}, Compos.~Math. \textbf{144} (2008), no.~2,
  503--521; {\tt arXiv: math/0601532 [math.QA]}.

\bibitem[BL00]{boli00}
{\sc L.A. Borisov and A.~Libgober}, \emph{Elliptic genera of toric varieties
  and applications to mirror symmetry}, Invent. Math. \textbf{140} (2000),
  no.~2, 453--485; {\tt arXiv: math/9904126 [math.AG]}.

\bibitem[BL03]{boli00a}
\leavevmode\vrule height 2pt depth -1.6pt width 23pt, \emph{Elliptic genera of
  singular varieties}, Duke Math.~J. \textbf{116} (2003), no.~2, 319--351; {\tt
  arXiv: math/0007108 [math.AG]}.

\bibitem[Bor01]{bo01}
{\sc L.A. Borisov}, \emph{Vertex algebras and mirror symmetry}, Commun. Math.
  Phys. \textbf{215} (2001), no.~2, 517--557; {\tt arXiv: math/9809094
  [math.AG]}.

\bibitem[BPH92]{bph92}
{\sc P.~Berglund, L.~Parkes, and T.~Hubsch}, \emph{{The Complete matter sector
  in a three generation compactification}}, Commun. Math. Phys. \textbf{148}
  (1992), 57--100.

\bibitem[CdGP91]{cogp91}
{\sc P.~Candelas, X.C. {de la Ossa}, P.S. Green, and L.~Parkes}, \emph{A pair
  of {C}alabi-{Y}au manifolds as an exactly soluble superconformal theory},
  Nucl. Phys. \textbf{B359} (1991), 21--74.

\bibitem[Cec91]{ce91}
{\sc S.~Cecotti}, \emph{${N}=2$ {L}andau-{G}inzburg vs. {C}alabi-{Y}au
  $\sigma$-models: {N}on--per\-tur\-ba\-tive aspects}, Int. J. Mod. Phys.
  \textbf{A6} (1991), 1749--1813.

\bibitem[CENT85]{cent85}
{\sc A.~Casher, F.~Englert, H.~Nicolai, and A.~Taormina}, \emph{Consistent
  superstrings as solutions of the $D=26$ bosonic string theory}, Phys. Lett.
  \textbf{B162} (1985), 121--126.

\bibitem[CH14]{crho13}
{\sc Th. Creutzig and G.~H{\"o}hn}, \emph{{Mathieu Moonshine and the Geometry
  of K3 Surfaces}}, Commun. Number Theory Phys. \textbf{8} (2014), no.~2,
  295--328; {\tt arXiv: 1309.2671 [math.QA]}.

\bibitem[Che10]{ch10}
{\sc M.C.N. Cheng}, \emph{{K3} surfaces, {$N=4$} dyons, and the {M}athieu group
  {$M_{24}$}}, Commun. Number Theory Phys. \textbf{4} (2010), 623--657; {\tt
  arXiv: 1005.5415 [hep-th]}.

\bibitem[CLS90]{cls90}
{\sc P.~Candelas, M.~Lynker, and R.~Schimmrigk}, \emph{Calabi--{Y}au manifolds
  in weighted {\rm P(4)}}, Nucl. Phys. \textbf{B341} (1990), 383--402.

\bibitem[CW15]{crwe15}
{\sc Th. Creutzig and K.~Wendland}, \emph{in preparation}.

\bibitem[DG88]{digr88}
{\sc J.~Distler and B.~Greene}, \emph{Some exact results on the superpotential
  from {C}alabi-{Y}au compactifications}, Nucl. Phys. \textbf{B309} (1988),
  295--316.

\bibitem[DHVW85]{dhvw85}
{\sc L.J. Dixon, J.~Harvey, C.~Vafa, and E.~Witten}, \emph{Strings on
  orbifolds}, Nucl. Phys. \textbf{B261} (1985), 678--686.

\bibitem[DHVW86]{dhvw86}
\leavevmode\vrule height 2pt depth -1.6pt width 23pt, \emph{Strings on
  orbifolds {II}}, Nucl. Phys. \textbf{B274} (1986), 285--314.

\bibitem[Dix88]{di87}
{\sc L.J. Dixon}, \emph{Some world-sheet properties of superstring
  compactifications, on orbifolds and otherwise}, in: Superstrings, unified
  theories and cosmology 1987 ({T}rieste, 1987), vol.~4 of ICTP Ser. Theoret.
  Phys., World Sci. Publ., Teaneck, NJ, 1988, pp.~67--126.

\bibitem[DY93]{dfya93}
{\sc P.~{Di Francesco} and S.~Yankielowicz}, \emph{Ramond sector characters and
  ${N}=2$ {L}andau-{G}inzburg models}, Nucl. Phys. \textbf{B409} (1993),
  186--210; {\tt arXiv: hep-th/9305037}.

\bibitem[EH09]{eghi09}
{\sc T.~Eguchi and K.~Hikami}, \emph{{Superconformal Algebras and Mock Theta
  Functions 2. Rademacher Expansion for K3 Surface}}, Commun. Number Theory
  Phys. \textbf{3} (2009), 531--554; {\tt arXiv: 0904.0911 [math-ph]}.

\bibitem[EH11]{eghi11}
\leavevmode\vrule height 2pt depth -1.6pt width 23pt, \emph{{Note on Twisted
  Elliptic Genus of K3 Surface}}, Phys. Lett. \textbf{B694} (2011), 446--455;
  {\tt arXiv: 1008.4924 [hep-th]}.

\bibitem[EOT11]{eot10}
{\sc T.~Eguchi, H.~Ooguri, and Y.~Tachikawa}, \emph{Notes on the {$K3$} surface
  and the {M}athieu group {$M_{24}$}}, Exp. Math. \textbf{20} (2011), no.~1,
  91--96; {\tt arXiv: 1004.0956 [hep-th]}.

\bibitem[EOTY89]{eoty89}
{\sc T.~Eguchi, H.~Ooguri, A.~Taormina, and S.-K. Yang}, \emph{Superconformal
  algebras and string compactification on manifolds with ${SU}(n)$ holonomy},
  Nucl. Phys. \textbf{B315} (1989), 193--221.

\bibitem[ET87]{egta87}
{\sc T.~Eguchi and A.~Taormina}, \emph{Unitary representations of the ${N}=4$
  superconformal algebra}, Phys. Lett. \textbf{B196} (1987), 75--81.

\bibitem[ET88a]{egta88a}
\leavevmode\vrule height 2pt depth -1.6pt width 23pt, \emph{Character formulas
  for the ${N}=4$ superconformal algebra}, Phys. Lett. \textbf{B200} (1988),
  315--322.

\bibitem[ET88b]{egta88}
\leavevmode\vrule height 2pt depth -1.6pt width 23pt, \emph{Extended
  superconformal algebras and string compactifications}, Trieste School 1988:
  Superstrings, pp.~167--188.

\bibitem[ET88c]{egta88b}
\leavevmode\vrule height 2pt depth -1.6pt width 23pt, \emph{On the unitary
  representations of ${N}=2$ and ${N}=4$ superconformal algebras}, Phys. Lett.
  \textbf{210} (1988), 125--132.

\bibitem[EY90]{egya90}
{\sc T.~Eguchi and S.-K. Yang}, \emph{{N}=2 superconformal models as
  topological field theories}, Mod. Phys. Lett. \textbf{A5} (1990), 1693--1701.

\bibitem[FS07]{frsz07}
{\sc E.~Frenkel and M.~Szczesny}, \emph{Chiral de {R}ham complex and
  orbifolds}, J.~Algebr.~Geom. \textbf{16} (2007), no.~4, 599--624; {\tt arXiv:
  math/0307181 [math.AG]}.

\bibitem[Gan16]{ga12}
{\sc T.~Gannon}, \emph{Much ado about {M}athieu}, Adv. Math. \textbf{301}
  (2016), 322--358; {\tt arXiv: 1211.5531 [math.RT]}.

\bibitem[Get83]{ge83}
{\sc E.~Getzler}, \emph{Pseudodifferential operators on supermanifolds and the
  {A}tiyah-{S}inger index theorem}, Commun. Math. Phys. \textbf{92} (1983),
  no.~2, 163--178.

\bibitem[GHV10a]{ghv10b}
{\sc M.R. Gaberdiel, S.~Hohenegger, and R.~Volpato}, \emph{Mathieu moonshine in
  the elliptic genus of {K}3}, JHEP \textbf{1010} (2010), 062; {\tt arXiv:
  1008.3778 [hep-th]}.

\bibitem[GHV10b]{ghv10a}
\leavevmode\vrule height 2pt depth -1.6pt width 23pt, \emph{Mathieu twining
  characters for {K}3}, JHEP \textbf{1009} (2010), 058; {\tt arXiv: 1006.0221
  [hep-th]}.

\bibitem[GHV12]{ghv12}
\leavevmode\vrule height 2pt depth -1.6pt width 23pt, \emph{Symmetries of {K}3
  sigma models}, Commun. Number Theory Phys. \textbf{6} (2012), 1--50; {\tt
  arXiv: 1106.4315 [hep-th]}.

\bibitem[Gil73]{gi73}
{\sc P.B. Gilkey}, \emph{Curvature and the eigenvalues of the {D}olbeault
  complex for {K}aehler manifolds}, Adv. Math. \textbf{11} (1973), 311--325.

\bibitem[GKP17]{gakepa16}
{\sc M.R. Gaberdiel, Ch. Keller, and H.~Paul}, \emph{{Mathieu Moonshine and
  Symmetry Surfing}}, J.~Physics \textbf{A50} (2017), no.~47, 474002; {\tt
  arXiv: 1609.09302 [hep-th]}.

\bibitem[GM04]{goma03}
{\sc V.~Gorbounov and F.~Malikov}, \emph{Vertex algebras and the
  {L}andau-{G}inzburg/{C}alabi-{Y}au correspondence}, Mosc. Math. J. \textbf{4}
  (2004), no.~3, 729--779; {\tt arXiv: math.AG/0308114}.

\bibitem[GP90]{grpl90}
{\sc B.R. Greene and M.R. Plesser}, \emph{Duality in {C}alabi-{Y}au moduli
  space}, Nucl. Phys. \textbf{B338} (1990), 15--37.

\bibitem[Gre97]{gr97}
{\sc B.R. Greene}, \emph{String theory on {C}alabi-{Y}au manifolds}, in:
  Fields, strings and duality (Boulder, CO, 1996), World Sci. Publishing, River
  Edge, NJ, 1997, pp.~543--726; {\tt arXiv: hep-th/9702155}.

\bibitem[Gri16]{gr16}
{\sc F.~Grimm}, \emph{The chiral de Rham complex of tori and orbifolds}, Ph.D.
  thesis, Albert-Ludwigs-Universit{\"a}t Freiburg, 2016, supervisor:
  {K.~Wendland}; available at
  http://home.mathematik.uni-freiburg.de/mathphys/mitarbeiter/wendland/GrimmPhD.pdf.

\bibitem[Hae90]{ha90}
{\sc A.~Haefliger}, \emph{Orbi-espaces}, in: {Sur les groupes hyperboliques
  d'apr\`es Mikhael Gromov}, Etienne {Ghys} and Pierre {de la Harpe}, eds.,
  Boston, MA: Birkh{\"a}user, 1990, pp.~203--213.

\bibitem[Hel09]{he09}
{\sc R.~Heluani}, \emph{Supersymmetry of the chiral de {R}ham complex. {II}.
  {C}ommuting sectors}, Internat. Math. Res. Notices (2009), no.~6, 953--987;
  {\tt arXiv: 0806.1021 [math.QA]}.

\bibitem[Hir88]{hi88}
{\sc F.~Hirzebruch}, \emph{Elliptic genera of level {$N$} for complex
  manifolds}, in: Differential geometric methods in theoretical physics,
  K.~Bleuler and M.~Werner, eds., Kluwer Acad. Publ., 1988, pp.~37--63.

\bibitem[{Huy}95]{huy95}
{\sc D.~{Huybrechts}}, \emph{{The tangent bundle of a Calabi-Yau manifold.
  Deformations and restriction to rational curves}}, Commun. Math. Phys.
  \textbf{171} (1995), no.~1, 139--158.

\bibitem[Kap05]{ka05}
{\sc A.~Kapustin}, \emph{Chiral de {R}ham complex and the half-twisted
  sigma-model}; {\tt arXiv: hep-th/0504074}.

\bibitem[Kon95]{ko95}
{\sc M.~Kontsevich}, \emph{Homological algebra of mirror symmetry}, Proceedings
  of the International Congress of Mathematicians, Vol.\ 1, 2 (Z\"urich, 1994)
  (Basel), Birkh{\"a}user, pp.~120--139; {\tt arXiv: alg-geom/9411018}.

\bibitem[Kri90]{kr90}
{\sc I.~Krichever}, \emph{Generalized elliptic genera and {B}aker-{A}khiezer
  functions}, Math. Notes \textbf{47} (1990), 132--142.

\bibitem[KT17]{katr16}
{\sc S.~Kachru and A.~Tripathy}, \emph{The Hodge-elliptic genus, spinning BPS
  states, and black holes}, Commun. Math. Phys. \textbf{355} (2017), 245--259;
  {\tt arXiv: 1609.02158 [hep-th]}.

\bibitem[KT18]{katr17}
\leavevmode\vrule height 2pt depth -1.6pt width 23pt, \emph{BPS jumping loci
  and special cycles}; {\tt arXiv: 1703.00455 [hep-th]}.

\bibitem[LL07]{lili07}
{\sc B.H. {Lian} and A.R. {Linshaw}}, \emph{{Chiral equivariant cohomology.
  I.}}, Adv. Math. \textbf{209} (2007), no.~1, 99--161; {\tt arXiv:
  math/0501084 [math.DG]}.

\bibitem[LVW89]{lvw89}
{\sc W.~Lerche, C.~Vafa, and N.P. Warner}, \emph{Chiral rings in ${N}=2$
  superconformal theories}, Nucl. Phys. \textbf{B324} (1989), 427--474.

\bibitem[Mor33]{mo33}
{\sc L.J. Mordell}, \emph{The definite integral
  $\int_{-\infty}^\infty{e^{ax^2+b}\over e^{cx}+d}dx$ and analytic theory of
  numbers}, Acta Math. \textbf{61} (1933), 323--360.

\bibitem[MS99]{masc99}
{\sc F.~Malikov and V.~Schechtman}, \emph{Chiral {P}oincar{\'e} duality},
  Math.~Res.~Lett. \textbf{6} (1999), no.~5-6, 533--546; {\tt arXiv:
  math/9905008 [math.AG]}.

\bibitem[MSV99]{msv98}
{\sc F.~Malikov, V.~Schechtman, and A.~Vaintrob}, \emph{Chiral de {R}ham
  complex}, Commun. Math. Phys. \textbf{204} (1999), no.~2, 439--473; {\tt
  arXiv: math/9803041 [math.AG]}.

\bibitem[Muk88]{mu88}
{\sc S.~Mukai}, \emph{Fi\-nite groups of auto\-mor\-phisms of {K}3 sur\-fa\-ces
  and the {M}athieu group}, Invent. Math. \textbf{94} (1988), 183--221.

\bibitem[Nar86]{na86}
{\sc K.S. Narain}, \emph{New heterotic string theories in uncompactified
  dimensions $< 10$}, Phys. Lett. \textbf{169B} (1986), 41--46.

\bibitem[NW01]{nawe00}
{\sc W.~Nahm and K.~Wendland}, \emph{A hiker's guide to {K}3 -- Aspects of
  ${N}=(4,4)$ superconformal field theory with central charge $c=6$}, Commun.
  Math. Phys. \textbf{216} (2001), 85--138; {\tt arXiv: hep-th/9912067}.

\bibitem[Oog89]{oo89}
{\sc H.~Ooguri}, \emph{{Superconformal Symmetry and Geometry of Ricci Flat
  Kahler Manifolds}}, Int. J. Mod. Phys. \textbf{A4} (1989), 4303--4324.

\bibitem[Pat71]{pa71}
{\sc V.K. Patodi}, \emph{An analytic proof of {R}iemann-{R}och-{H}irzebruch
  theorem for {K}aehler manifolds}, J. of Diff. Geometry \textbf{5} (1971),
  251--283.

\bibitem[Sei88]{se88}
{\sc N.~Seiberg}, \emph{Observations on the moduli space of superconformal
  field theories}, Nucl. Phys. \textbf{B303} (1988), 286--304.

\bibitem[Sen86]{se86}
{\sc A.~Sen}, \emph{$(2,0)$ supersymmetry and space-time supersymmetry in the
  heterotic string theory}, Nucl. Phys. \textbf{B278} (1986), 289--308.

\bibitem[Sen87]{se87}
\leavevmode\vrule height 2pt depth -1.6pt width 23pt, \emph{Heterotic string
  theory on {C}alabi-{Y}au manifolds in the Green-Schwarz formalism}, Nucl.
  Phys. \textbf{B284} (1987), 423--448.

\bibitem[Son16]{so16}
{\sc Bailin Song}, \emph{Vector bundles induced from jet schemes}; {\tt arXiv:
  1609.03688 [math.DG]}.

\bibitem[Son17]{so17}
\leavevmode\vrule height 2pt depth -1.6pt width 23pt, \emph{Chiral {H}odge
  cohomology and {M}athieu moonshine}; {\tt arXiv: 1705.04060 [math.QA]}.

\bibitem[Tao90]{ta89}
{\sc A.~Taormina}, \emph{The ${N}=2$ and ${N}=4$ superconformal algebras and
  string com\-pac\-ti\-fi\-ca\-tions}, Mathematical physics (Islamabad, 1989),
  World Sci. Publishing, pp.~349--370.

\bibitem[Thu97]{th97}
{\sc W.P. Thurston}, \emph{Three-dimensional geometry and topology. {V}ol. 1},
  Princeton University Press, Princeton, NJ, 1997, ed. Silvio Levy.

\bibitem[TW13]{tawe11}
{\sc A.~Taormina and K.~Wendland}, \emph{The overarching finite symmetry group
  of {K}ummer surfaces in the {M}athieu group {$M_{24}$}}, JHEP \textbf{08}
  (2013), 125; {\tt arXiv: 1107.3834 [hep-th]}.

\bibitem[TW15a]{tawe13}
\leavevmode\vrule height 2pt depth -1.6pt width 23pt, \emph{Symmetry-surfing
  the moduli space of {K}ummer {K}3s}, Proc. of the Conference String-Math
  2012, Proceedings of Symposia in Pure Mathematics, no.~90, pp.~129--153; {\tt
  arXiv: 1303.2931 [hep-th]}.

\bibitem[TW15b]{tawe12}
\leavevmode\vrule height 2pt depth -1.6pt width 23pt, \emph{A twist in the
  {$M_{24}$} moonshine story}, Confluentes Mathematici \textbf{7} (2015),
  83--113; {\tt arXiv: 1303.3221 [hep-th]}.

\bibitem[TW17]{tawe17}
\leavevmode\vrule height 2pt depth -1.6pt width 23pt, \emph{The Conway
  Moonshine Module is a Reflected K3 Theory}; {\tt arXiv: 1704.03813 [hep-th]}.

\bibitem[Wen00]{diss}
{\sc K.~Wendland}, \emph{Moduli spaces of unitary conformal field theories},
  Ph.D. thesis, University of Bonn, 2000, supervisor: {W.~Nahm}; available on
  request.

\bibitem[Wen15]{we14}
\leavevmode\vrule height 2pt depth -1.6pt width 23pt, \emph{Snapshots of
  conformal field theory}, in: Mathematical Aspects of Quantum Field Theories,
  D.~Calaque and Th. Strobl, eds., Springer-Verlag, 2015, pp.~89--129; {\tt
  arXiv: 1404.3108 [hep-th]}.

\bibitem[Wen17]{we15}
\leavevmode\vrule height 2pt depth -1.6pt width 23pt, \emph{K3 en route From
  Geometry to Conformal Field Theory}, Proceedings of the 2013 Summer School
  ``Geometric, Algebraic and Topological Methods for Quantum Field Theory{"} in
  Villa de Leyva, Colombia, World Scientific, pp.~75--110; {\tt arXiv:
  1503.08426 [math.DG]}.

\bibitem[Wit82]{wi82}
{\sc E.~Witten}, \emph{Constraints on supersymmetry breaking}, Nucl. Phys.
  \textbf{B202} (1982), 253--316.

\bibitem[Wit87]{wi87}
\leavevmode\vrule height 2pt depth -1.6pt width 23pt, \emph{Elliptic genera and
  quantum field theory}, Commun. Math. Phys. \textbf{109} (1987), 525--536.

\bibitem[Wit88]{wi88}
\leavevmode\vrule height 2pt depth -1.6pt width 23pt, \emph{The index of the
  {D}irac operator in loop space}, in: Elliptic curves and modular forms in
  algebraic geometry, P.~Landweber, ed., Springer-Verlag, 1988, pp.~161--181.

\bibitem[Wit92]{wi91}
\leavevmode\vrule height 2pt depth -1.6pt width 23pt, \emph{Mirror manifolds
  and topological field theory}, in: Essays on mirror manifolds, Internat.
  Press, Hong Kong, 1992, pp.~120--158; {\tt arXiv: hep-th/9112056}.

\bibitem[Wit94]{wi94}
\leavevmode\vrule height 2pt depth -1.6pt width 23pt, \emph{On the
  {L}andau-{G}inzburg description of ${N}=2$ minimal models}, Int. J. Mod.
  Phys. \textbf{A9} (1994), 4783--4800; {\tt arXiv: hep-th/9304026}.

\end{thebibliography}
%
\newcommand{\etalchar}[1]{$^{#1}$}
\def\polhk#1{\setbox0=\hbox{#1}{\ooalign{\hidewidth
  \lower1.5ex\hbox{`}\hidewidth\crcr\unhbox0}}} \def\cprime{$^\prime$}
  \newcommand{\noopsort}[1]{}
\providecommand{\bysame}{\leavevmode\hbox to3em{\hrulefill}\thinspace}

\end{document}